\documentclass[sigconf]{aamas} 

\renewcommand\footnotetextcopyrightpermission[1]{}

\usepackage{balance} 
\usepackage{algorithm}
\usepackage[noend]{algpseudocode}
\usepackage{mathtools}
\usepackage{amsmath}
\usepackage{tikz}
\usetikzlibrary{shapes, shapes.symbols, backgrounds, matrix, calc, arrows, math, arrows.meta, decorations.pathreplacing, decorations.markings, patterns}
\usepackage{todonotes}
\usepackage{dsfont}
\usepackage{enumitem}
\usepackage{tikz}
\usepackage{subfigure}
\usepackage{amsthm}
\usepackage{thm-restate}
\usepackage{booktabs}
\usepackage{multirow}
\usepackage{newfloat}
\usepackage{wrapfig}
\usepackage{dsfont}
\usepackage{comment}
\usepackage{multirow}
\usepackage{graphicx}
\usepackage{thm-restate}

\theoremstyle{plain}
\newtheorem{theorem}{Theorem}[section]
\newtheorem{lemma}[theorem]{Lemma}
\newtheorem{proposition}[theorem]{Proposition}
\newtheorem{corollary}[theorem]{Corollary}

\theoremstyle{definition}
\newtheorem{definition}{Definition}[section]

\theoremstyle{remark}

\acmConference[ArXiV]{ArXiV}{October}{2022}
\copyrightyear{}
\acmDOI{}
\acmPrice{}
\acmISBN{}

\acmSubmissionID{126}

\newcommand{\R}{\mathbb{R}}

\title{Fairly Allocating (Contiguous) Dynamic Indivisible Items with Few Adjustments}

\author{Mingwei Yang}
\affiliation{
  \institution{Peking University}
  \city{Beijing}
  \country{China}}
\email{yangmingwei@pku.edu.cn}

\begin{abstract}

We study the problem of dynamically allocating $T$ indivisible items to $n$ agents with the restriction that the allocation is fair all the time.
Due to the negative results to achieve fairness when allocations are irrevocable, we allow adjustments to make fairness attainable with the objective to minimize the number of adjustments.
For restricted additive or general identical valuations, we show that \textit{envy-freeness up to one item (EF1)} can be achieved with no adjustments.
For additive valuations, we give an EF1 algorithm that requires $O(mT)$ adjustments, improving the previous result of $O(nmT)$ adjustments, where $m$ is the maximum number of different valuations for items among all agents.

We further impose the contiguity constraint on items such that items are arranged on a line by the order they arrive and require that each agent obtains a consecutive block of items.
We present extensive results to achieve either \textit{proportionality} with an additive approximate factor (PROPa) or EF1, where PROPa is a weaker fairness notion than EF1.
In particular, we show that for identical valuations, achieving PROPa requires $\Theta(nT)$ adjustments.
Moreover, we show that it is hopeless to make any significant improvement for either PROPa or EF1 when valuations are nonidentical.

Our results exhibit a large discrepancy between the identical and nonidentical cases in both contiguous and noncontiguous settings.
All our positive results are computationally efficient.

\end{abstract}


\newcommand{\BibTeX}{\rm B\kern-.05em{\sc i\kern-.025em b}\kern-.08em\TeX}

\begin{document}




\maketitle 


\section{Introduction}

Fair division is one of the most fundamental and well-studied topics in Computational Social Choice with much significance and several applications in many real-life scenarios \cite{DBLP:journals/jasss/Seidl18, DBLP:journals/sigecom/GoldmanP14}.
Generally, there are some resources and our objective is to divide them among a group of competing agents in a fair manner.
In our discussion, we assume that the items to be allocated are \textit{goods}, whose valuations are nonnegative.
Arguably, the most compelling fairness notion is \textit{envy-freeness}, which is defined as each agent weakly preferring his own bundle to any other agent's bundle.

However, in the indivisible regime, the existence of envy-free solutions is not guaranteed.
For instance, if there are two agents but only one item, the agent who receives the item is certainly envied by the other one.
One of the natural relaxations of envy-freeness is \textit{envy-freeness up to one item (EF1)} \cite{DBLP:conf/bqgt/Budish10}.
In an EF1 allocation, every agent $i$ may envy another agent $j$, but the envy could be eliminated by removing one item from agent $j$'s bundle.
EF1 allocations are always guaranteed to exist and can be computed in polynomial time even for general valuations \cite{DBLP:conf/sigecom/LiptonMMS04}.

Imposing some constraints to the model will make the fairness objective less tractable.
A series of works focus on the setting where items lie on a line and each agent obtains a consecutive block of items.
For monotone valuations, Bilò et al. \cite{DBLP:journals/geb/BiloCFIMPVZ22} present polynomial-time algorithms to compute contiguous EF1 allocations for any number of agents with identical valuations or at most three agents, and prove the existence of contiguous EF1 allocations for four agents.
More recently, it is shown that contiguous EF1 allocations for any number of agents always exist \cite{igarashi2022cut}.

Another natural generalization assumes that items arrive online.
Furthermore, the types of future items are unknown and decisions have to be made instantly.
He et al. \cite{DBLP:conf/ijcai/HePPZ19} investigate this model with the requirement that the allocations returned after the arrival of each item are EF1 and consider additive valuations.
However, due to the negative results against adversaries when allocations are irrevocable \cite{DBLP:conf/sigecom/BenadeKPP18}, adjustments are necessary to achieve EF1 deterministically.
Notably, the $O(T^2)$ of adjustments is attained trivially by redistributing all items in each round, where $T$ is the number of items.
He et al. \cite{DBLP:conf/ijcai/HePPZ19} show that the $\Omega(T)$ of adjustments is inevitable for more than two agents, even if the information of all items is known upfront.
On the positive side, they give two algorithms with respectively $O(T^{1.5})$ and $O(nmT)$ adjustments, where $n$ is the number of agents and $m$ denotes the maximum number of distinct valuations for items among all agents.

In this work, we adopt the model proposed by He et al. \cite{DBLP:conf/ijcai/HePPZ19}.
Furthermore, we also impose the contiguity constraint on items.
More precisely, items are arranged on a line by the order they arrive and each agent obtains a contiguous block of items.
To motivate this, consider a library that has several bookshelves with books of the same types being placed together.
Moreover, the numbers of books of two different types should not differ by a large amount.
With more and more bookshelves being deployed, the library needs to reallocate a consecutive block of bookshelves to each certain type of book and redistribute the books according to the new allocation.
In this case, the objective is to minimize the cost of moving the books or equivalently, the number of adjustments.
We study the number of adjustments necessary to achieve some fairness guarantee in both contiguous and noncontiguous settings.

\subsection{Our Contributions}


We first describe our positive results in the noncontiguous setting.
We show that if valuations are limited to be restricted additive\footnote{The valuations are restricted additive if they are additive and every item has an inherent valuation with every agent being interested in only some items \cite{DBLP:conf/ijcai/AkramiRS22}.} or general identical, EF1 can be achieved with no adjustments.
In addition, we give an EF1 algorithm for additive valuations that requires $O(mT)$ adjustments, improving the previous result of $O(nmT)$ adjustments.
Note that if $m$ is a constant, it matches the $\Omega(T)$ lower bound and thus is optimal.

With the contiguity constraint, EF1 is too stringent and we start with a weaker fairness notion.
It is known that contiguous $(\frac{n-1}{n} \cdot v^{\max})$-\textit{proportional (PROPa)} allocations can be computed efficiently, where $v^{\max}$ is the maximum valuation of items, and this additive approximate factor is tight in some sense \cite{DBLP:journals/dam/Suksompong19}.
We first consider identical valuations.
We give a PROPa algorithm that requires $O(nT)$ adjustments and then establish the matching lower bound to show that our algorithm is optimal.
When it comes to EF1, for two agents with general valuations, we show that EF1 is achievable with $O(T)$ adjustments, which is optimal.
If the valuation of each item is assumed to lie in $[L, R]$ such that $0 < L \leq R$, we give an EF1 algorithm that requires $O\left( (R / L) \cdot n^2 T\right)$ adjustments.
By contrast, in the nonidentical case, we show that it is hopeless to make any significant improvement even for additive valuations.
Specifically, we give instances to establish the lower bounds of $\Omega(T^2 / n)$ to achieve PROPa and $\Omega(T^2)$ for EF1.
Our results in the contiguous setting with additive valuations are summarized in Table~\ref{tab:results-contiguous}.

Our results exhibit a large discrepancy between the identical and nonidentical cases in both contiguous and noncontiguous settings.
In addition, all the algorithms given in this paper can be implemented in polynomial time.

\begin{table}[t]
\centering
\caption{Results for the contiguous setting with additive valuations.
The result for two agents in the identical case also works for general valuations.
}
\label{tab:results-contiguous}
\begin{tabular}{@{}lllllll@{}}
\toprule
        & Identical                        &                                  &                                                             &                      & \multicolumn{2}{l}{Nonidentical}                                         \\ \cmidrule(lr){2-4} \cmidrule(l){6-7} 
        & PROPa                            & \multicolumn{2}{l}{EF1}                                                                        &                      & PROPa                                 & EF1                               \\ \cmidrule(lr){3-4}
        &                                  & Lower                            & Upper                                                       &                      &                                       &                                   \\ \cmidrule(r){1-4} \cmidrule(l){6-7} 
$n = 2$ & \multicolumn{3}{c}{$\Theta(T)$}                                                                                                   & \multicolumn{1}{c}{} & \multicolumn{2}{c}{$\Theta(T^2)$}                                         \\
$n > 2$ & \multicolumn{1}{c}{$\Theta(nT)$} & \multicolumn{1}{c}{$\Omega(nT)$} & \multicolumn{1}{c}{$O\left(\frac{R}{L} \cdot n^2 T\right)$} & \multicolumn{1}{c}{} & \multicolumn{1}{c}{$\Omega(T^2 / n)$} & \multicolumn{1}{c}{$\Theta(T^2)$} \\ \bottomrule
\end{tabular}
\end{table}

\subsection{Related Work}

Even though both divisible and indivisible models are extensively studied, we only focus on the indivisible setting, which is more relevant to our work.

\paragraph{Dynamic fair division.}
Our work belongs to the vast literature of \textit{dynamic} or \textit{online} fair division \cite{DBLP:conf/aaai/AleksandrovW20}.
Under the assumption that valuations are normalized to $[0, 1]$ and items are allocated irrevocably, the maximum envy of $\tilde{O}(\sqrt{T / n})$ can be achieved deterministically and this bound is tight asymptotically \cite{DBLP:conf/sigecom/BenadeKPP18}.
To bypass the negative results, motivated by the notion of \textit{disruptions} \cite{DBLP:conf/sigecom/FriedmanPV15, DBLP:conf/sigecom/FriedmanPV17}, He et al. \cite{DBLP:conf/ijcai/HePPZ19} introduce adjustments to achieve EF1 deterministically.

Without the ability to adjust, another popular measure of compromise is assuming that agents' valuations are stochastic.
When the valuation of each agent to each item is drawn i.i.d. from some distribution, the algorithm of allocating each item to the agent with the maximum valuation for it is envy-free with high probability and ex-post Pareto optimal \cite{DBLP:journals/siamdm/ManurangsiS20, DBLP:journals/siamdm/ManurangsiS21}.
Besides, Bai et al. \cite{DBLP:conf/ijcai/BaiG22} provide the same guarantee for asymmetric agents, i.e., the valuation of an item for each agent is independently drawn from an agent-specific distribution.
Furthermore, assuming that the valuations of different agents for the same item are correlated, Pareto efficiency and fairness are also compatible \cite{DBLP:conf/sigecom/ZengP20}.
More recently, Benadè et al. \cite{benadedynamic} study the partial-information setting where only the ordinal information is revealed.
Another series of works \cite{DBLP:conf/ki/AleksandrovW17, DBLP:conf/pricai/AleksandrovW19} resort to random allocations to achieve ex-ante fairness together with some efficiency and incentive guarantees.

\paragraph{Fair division of contiguous blocks.}
The online fair division model with the contiguity requirement concerned in our work is a strict extension of the fair division of contiguous blocks problem.
In the offline setting, the existence of contiguous EF1 allocations is intensively studied \cite{DBLP:journals/geb/BiloCFIMPVZ22, DBLP:journals/siamdm/OhPS21, igarashi2022cut} and the approximate versions of proportionality, envy-freeness, and equitability are also considered \cite{DBLP:journals/dam/Suksompong19}.
More recently, Misra et al. \cite{DBLP:journals/corr/abs-2101-09794} designs an algorithm to compute a contiguous EQ1 allocation with the egalitarian welfare guarantee.
Besides, the \textit{price of fairness} of contiguous allocations for goods and chores are respectively established by Suksompong \cite{DBLP:journals/dam/Suksompong19} and Höhne et al. \cite{DBLP:journals/iandc/HohneS21}.
More generally, the connectivity relation among items is allowed to form a graph that possesses some structures with a path being a special case \cite{DBLP:conf/ijcai/BouveretCEIP17, DBLP:journals/aamas/BouveretCL19}.
\section{Preliminaries}
There are $n$ agents and let $N = [n]$ be the set of all agents.
There are $T$ items arriving one by one.
Let $g_t$ be the $t$-th item and $M_t = \{g_1, \ldots, g_t\}$ be the set of the first $t$ items that arrive.
In particular, define $M = M_T$.
Let $M_{l, r} = \{g_{l + 1}, \ldots, g_r\}$ be the set of items whose indexes are in $[l + 1, r]$ with $M_{0, r} = M_r$.
Each agent $i$ has a nonnegative valuation $v_i : 2^M \to \R_+$, which reflects his evaluation for each subset of items.
Denote $m = \max_{i \in N} |\{v_i(g) \mid g \in M\}|$ as the maximum number of distinct valuations for items among all agents.
To simplify the notations, for $S \subseteq M$ and $g \in M$, we use $v_i(g)$, $v_i(S + g)$ and $v_i(S - g)$ to respectively represent $v_i(\{g\})$, $v_i(S \cup \{g\})$ and $v_i(S \setminus \{g\})$.
An allocation for a set of items $M'$ is a tuple $(A_1, A_2, \ldots, A_n)$ satisfying $A_i \cap A_j = \emptyset$ for all $i \neq j$ and $A_1 \cup A_2 \cup \ldots \cup A_n = M'$, where $A_i$ is agent $i$'s bundle.

After the arrival of each item, algorithms should return an allocation for all the existing items.
We use the number of adjustments to measure the performance of algorithms.
Formally, let the allocations returned by an algorithm be $A^1, \ldots, A^T$, where $A^t$ is an allocation for $M^t$.
The number of adjustments in round $t$ is defined as the number of items $g \in M_{t - 1}$ such that $g$ belongs to different agents in $A^{t - 1}$ and $A^t$.
The number of adjustments required by an algorithm is defined as the total number of adjustments across all $T$ rounds.

Now we enumerate different types of valuation functions that are considered in this paper.
A valuation $v$ is \textit{additive} if for all $S \subseteq M$, we have $v(S) = \sum_{g \in S} v(g)$.
Valuations $v_1, \ldots, v_n$ are \textit{restricted additive} if they are additive and each item $g$ has an inherent valuation $v_g$ so that $v_i(g) \in \{0, v_g\}$ for every agent $i$.
A valuation $v$ is \textit{general} if it is only required to be monotone, i,e, $v(S + g) \geq v(S)$ for all $S \subseteq M$ and $g \in M$.
Valuations $v_1, \ldots, v_n$ are \textit{identical} if we have $v_1(\cdot) = \cdots = v_n(\cdot)$.
We assume valuations to be additive if without specification.


We formally define the fairness criteria including EF1 and PROPa in the following.


\begin{definition}[EF1]
    An allocation $A$ is \textit{envy-free up to one item (EF1)} if for any $i, j \in N$, either $v_i(A_i) \geq v_i(A_j)$ or there exists $g \in A_j$ such that $v_i(A_i) \geq v_i(A_j - g)$.
\end{definition}


\begin{definition}[PROPa]
    Let $v^{\max} = \max_{i \in N, g \in M} v_i(g)$.
    An allocation $A$ is $(\frac{n - 1}{n} \cdot v^{\max})$-\textit{proprtional (PROPa)} if we have
    \begin{align*}
        v_i(A_i) \geq \frac{1}{n} \cdot v_i(M) - \frac{n - 1}{n} \cdot v^{\max}
    \end{align*}
    for any $i \in N$.
\end{definition}

The following proposition shows that EF1 is a stronger notion than PROPa.

\begin{proposition}\label{prop:EF1-implies-PROPa}
    If an allocation $A = (A_1, \ldots, A_n)$ is EF1, then it is also PROPa.
\end{proposition}

\begin{proof}
    Suppose that there are $n$ agents with valuations $v_1, \ldots, v_n$ and allocation $A = (A_1, \ldots, A_n)$ is EF1.
    Fix $i \in N$.
    If $A_j$ is empty for some $j \in N$, then we add an item with valuation $0$ for each agent to $A_j$.
    From now on, we assume that $A_1, \ldots, A_n$ are nonempty.
    By the definition of EF1, for any $j \in N$, there exists $g_j \in A_j$ such that
    \begin{align*}
        v_i(A_i) \geq v_i(A_j - g_j).
    \end{align*}
    Summing over all $j \neq i$, we obtain
    \begin{align*}
        (n - 1) \cdot v_i(A_i) \geq \sum_{j \neq i} v_i(A_j - g_j)
    \end{align*}
    Adding $v_i(A_i)$ to both sides,
    \begin{align*}
        n \cdot v_i(A_i) \geq v_i(M) - \sum_{j \neq i} v_i(g_j).
    \end{align*}
    Therefore,
    \begin{align*}
        v_i(A_i)
        & \geq \frac{1}{n} \cdot v_i(M) - \frac{1}{n} \sum_{j \neq i} v_i(g_j)\\
        & \geq \frac{1}{n} \cdot v_i(M) - \frac{n-1}{n} \cdot v^{\max}
    \end{align*}
    as desired.
\end{proof}
\section{Noncontiguous Setting}

In this section, we consider the noncontiguous setting.
We first present greedy algorithms for restricted additive and general identical valuations to show that EF1 can be achieved at no cost in both cases.
Next, we give an algorithm that requires $O(mT)$ adjustments.
In Appendix~\ref{apx:two-agents-and-mixed-manna}, we show that when there are only two agents, we can achieve the optimal $O(T)$ adjustments even if the valuations are allowed to be negative.

\subsection{Restricted Additive or General Identical Valuations}

For additive valuations, recall that every EF1 algorithm requires at least $\Omega(T)$ adjustments \cite{DBLP:conf/ijcai/HePPZ19}.
We show that if we limit the valuations to restricted additive or general identical valuations, EF1 can be achieved with no cost: simple greedy algorithms is EF1 and require no adjustments for these two cases.
Notably, restricted additive valuations can be viewed as a slightly more general version of additive identical valuations by adding the assumption that each agent may be indifferent to some items.
As a consequence, there is a large discrepancy between the identical and nonidentical cases in the noncontiguous setting.

Firstly, we give a greedy algorithm for restricted additive valuations in Algorithm~\ref{greedy-algorithm-for-restricted-additive-agents}, whose idea is straightforward.
In each round, we simply give the new item to the agent with the minimum valuation for his current bundle among the agents that positively value the new item.
Without loss of generality, we assume that for every item, there exists an agent that has a positive valuation for it.
Since otherwise all agents are indifferent to the item and we just discard it.

\begin{algorithm}[!htp]
        \caption{Greedy Algorithm for Restricted Additive Valuations}
        \label{greedy-algorithm-for-restricted-additive-agents}
        \begin{algorithmic}[1]
            \Require {$v_i$ for each agent $a_i$}
            \State {$A = \emptyset$}
            \For {$t = 1, \ldots, T$}
                \State {$k \leftarrow \arg\min_{i:v_i(g_t) > 0} v_i(A_i)$}
                \State {$A_{k} \leftarrow A_{k} \cup \{g_t\}$}
                \State {$A^t \leftarrow A$}
            \EndFor
            \State {\Return $[A^1, A^2, \ldots, A^T]$}
        \end{algorithmic}
\end{algorithm}

\begin{theorem}
    Algorithm~\ref{greedy-algorithm-for-restricted-additive-agents} is an EF1 algorithm for restricted additive valuations that requires no adjustments.
\end{theorem}

\begin{proof}
    Algorithm~\ref{greedy-algorithm-for-restricted-additive-agents} makes no adjustments since once an item is assigned to an agent, it will be allocated to the same agent in the future.
    Now we prove that every returned allocation is EF1.
    First we prove the property that, for any $i, j \in N$, $v_i(A_j) \leq v_j(A_j)$ always holds.
    In other words, each bundle is valued the most by the agent who owns it.
    Since Algorithm~\ref{greedy-algorithm-for-restricted-additive-agents} only assigns an item to the agent who has positive valuation for it, for any agent $j$, we have $v_j(A_j) = \sum_{g \in A_j} v_g$, where $v_g$ is the inherent valuation of item $g$.
    Since agent $i$ may be indifferent to some items in $A_j$, we have $v_i(A_j) \leq \sum_{g \in A_j} v_g = v_j(A_j)$.
    
    Next, we prove that each allocation is EF1 by induction.
    Before the arrival of the first item, the allocation is trivially EF1.
    We will show that if the current allocation is EF1, then after the arrival of a new item, the newly generated allocation is also EF1.
    When a new item $g$ comes, define $N_1 = \{i \in N \mid v_i(g) > 0\}$ as the set of agents who are interested in $g$ and $N_2 = N \setminus N_1$ as the set of agents who are not.
    Since the agents in $N_2$ are indifferent to $g$, their presence cannot lead to the violation of EF1 no matter which agent $g$ is assigned to.
    Therefore, we are only concerned about the agents in $N_1$.
    Let $k \in N_1$ be the agent who gets $g$.
    By the choice of $k$, for any $i \in N_1$, we have $v_i(A_k) \leq v_k(A_k) \leq v_i(A_i)$ before the assignment of $g$, which means that each agent in $N_1$ does not envy agent $k$. Thus after allocating $g$ to agent $k$, each agent in $N_1$ will not envy agent $k$ up to one item.
    Therefore, the new allocation is also EF1.
\end{proof}

Then we give an algorithm for general identical valuations in Algorithm~\ref{greedy-algorithm-for-general-identical-agents}.
The allocating strategy of Algorithm~\ref{greedy-algorithm-for-general-identical-agents} is assigning the new item to the agent with the minimum valuation for his current bundle.

\begin{algorithm}[!htp]
        \caption{Greedy Algorithm for General Identical Valuations}
        \label{greedy-algorithm-for-general-identical-agents}
        \begin{algorithmic}[1]
            \Require {A valuation $v$ for all agents}
            \State {$A = \emptyset$}
            \For {$t = 1, \ldots, T$}
                \State {$k \leftarrow \arg\min_{i} v(A_i)$}
                \State {$A_k \leftarrow A_k \cup \{g_t\}$}
                \State {$A^t \leftarrow A$}
            \EndFor
            \State {\Return $[A^1, A^2, \ldots, A^T]$}
        \end{algorithmic}
\end{algorithm}

\begin{theorem}
    Algorithm~\ref{greedy-algorithm-for-general-identical-agents} is an EF1 algorithm for general identical valuations that requires no adjustments.
\end{theorem}

\begin{proof}
    Algorithm~\ref{greedy-algorithm-for-general-identical-agents} makes no adjustments since once an item is assigned to an agent, it will be allocated to the same agent in the future.
    Now we prove that every returned allocation is EF1 by induction.
    When there are no items, the allocation is trivially EF1.
    We will prove that if the current allocation is EF1, then after the arrival of a new item, the newly generated allocation is also EF1.
    When a new item $g$ comes, suppose that agent $k$ is the agent with the minimum valuation for his current bundle and thus gets $g$.
    Since all agents do not envy agent $k$ before the assignment by the choice of $k$, they will not envy agent $k$ up to one item in the new allocation.
    Therefore, the new allocation is also EF1.
\end{proof}

\subsection{Additive Valuations}

For additive valuations, we present an EF1 algorithm that requires $O(mT)$ adjustments in Algorithm~\ref{layer-updating-algorithm}.
Note that when $m$ is a constant, Algorithm~\ref{layer-updating-algorithm} is optimal.

The idea is based on the round-robin algorithm.
Recall that in the round-robin algorithm, agents pick their favorite items alternately in each round.
The resulting allocation is EF1 since, in each round, every agent prefers the item that he picks in this round to any item picked by some agent in the next round.
Based on this observation, we aim to maintain the above property for the allocation in each round\footnote{This approach is applied in \cite{DBLP:conf/sigecom/BrustleDNSV20} to show that the allocation induced by recursively finding a maximum-weight matching between agents and items is EF1.}.
To this end, we maintain a round-robin structure by dividing items into multiple layers satisfying that each agent gets exactly one item in each layer.
In particular, each agent gets no more than one item in the last layer.
Moreover, it holds that in each layer, each agent prefers his item in this layer to any item in the next layer.
In this case, the allocation induced by this structure is EF1.

Algorithm~\ref{layer-updating-algorithm} is the realization of the above idea.
We use $C^k$ to store the allocation in layer $k$, where $C^k_a$ denotes the item obtained by agent $a$.
In particular, let $C^i_a = 0$ if agent $i$ does not obtain any item in layer $i$, and let $v_i(0) = 0$ for all $i \in N$.
When a new item arrives, the algorithm updates each layer sequentially.
Whenever there is an agent who prefers the new item to his item in the current layer, we swap the new item with his item in the current layer (Line~\ref{algoline:replace1}--Line~\ref{algoline:replace2}).
After all the replacements in the current layer are finished, we move to the next layer and repeat the above process.
In the last layer, we assign the new item to an arbitrary agent who has not obtained any item in this layer (Line~\ref{algoline:last-layer}).
Finally, we return the allocation collectively induced by all layers (Line~\ref{algoline:final-allocation}).

\begin{algorithm}[!htp]
    \caption{Layer Updating Algorithm}
    \label{layer-updating-algorithm}
    \begin{algorithmic}[1]
        \Require {$v_i$ for each agent $a_i$}
        \For {$t = 1, \ldots, T$}
            \State {$k \leftarrow \lceil \frac{t}{n} \rceil$}
            \For {$i = 1, \ldots, k - 1$}
                \While {$G := \{a \in N \mid v_a(g_t) > v_a(C_a^i)\} \neq \emptyset$}\label{algoline:replace1}
                    \State {$\widehat{a} \leftarrow \arg\min_{a \in G} v_a(C_a^i)$}
                    \State {swap($g_t$, $C^i_{\widehat{a}}$)}\label{algoline:replace2}
                \EndWhile
            \EndFor
            \State {$C^{k}_{t - n * (k - 1)} \leftarrow g_t$}\label{algoline:last-layer}
            \State {Let $A^t$ be the combination of allocations $C^1, C^2, \ldots, C^k$}\label{algoline:final-allocation}
        \EndFor
        \State {\Return $[A^1, A^2, \ldots, A^T]$}
    \end{algorithmic}
\end{algorithm}

\begin{theorem}
    Algorithm~\ref{layer-updating-algorithm} is an EF1 algorithm that requires $O(mT)$ adjustments.
\end{theorem}

\begin{proof}
    We first prove that throughout the algorithm, for each layer $k$, each agent prefers the item that he obtains in this layer to any item in layer $k + 1$, i.e., $v_i(C^k_i) \geq v_i(C^{k+1}_j)$ for any $i, j \in N$.
    Then we show that this condition is sufficient for EF1.
    
    Now we use induction to prove that
    \begin{align}
        v_i(C_i^k) \geq v_i(C_j^{k+1}), \quad \forall i, j \in N, \forall k > 0.\label{eqn:round-robin-property}
    \end{align}
    When there are no items, \eqref{eqn:round-robin-property} is trivially satisfied.
    We will show that if \eqref{eqn:round-robin-property} holds for the current allocation, then after the arrival of a new item, \eqref{eqn:round-robin-property} also holds for the newly generated allocation.
    Note that during the update of a layer, each replacement will strictly increase the valuation of the chosen agent for his item in this layer.
    Thus after each replacement, the chosen agent still prefers his item in this layer to any item in the next layer.
    Moreover, the update for a layer being finished indicates that each agent prefers his item in this layer to the new item.
    As a result, we can safely move on and update the next layer with the new item.
    
    We show that \eqref{eqn:round-robin-property} implies EF1.
    Suppose that $(C^1, C^2, \ldots)$ satisfies \eqref{eqn:round-robin-property} and let $A$ be the allocation induced by $(C^1, C^2, \ldots)$.
    Note that for any $i \in N$, we have $A_i = \{C^1_i, C^2_i, \ldots\}$.
    Thus for any agents $i, j \in N$, we have
    \begin{align*}
        v_i(A_i) = \sum_{k \geq 1} v_i(C^k_i)
        \geq \sum_{k \geq 1} v_i(C^{k+1}_j)
        = v_i(A_j - C_j^1),
    \end{align*}
    which means that agent $i$ does not envy agent $j$ up to one item.
    Therefore, $A$ is EF1.
    
    Finally, we prove that the algorithm uses $O(mT)$ adjustments in total.
    Notice that the number of adjustments is upper bounded by the number of replacements.
    Since a replacement will make the valuation of the chosen agent for his item in the current layer strictly larger and there are at most $m$ distinct valuations of each agent for items, the number of replacements for each agent in a certain layer is at most $m$.
    Therefore, the number of adjustments is $\lceil T / n \rceil \cdot n \cdot m = O(mT)$ since there are $\lceil T / n \rceil$ layers.
\end{proof}

We also notice that all the positive results in the noncontiguous setting share the property that the valuation of each agent will not decrease across the entire period.
In other words, agents will only be better off if they follow the dictation of adjustments and, as a result, they have the incentive to participate.
By contrast, in the contiguous setting, the incentive property is not compatible with EF1.
See Appendix~\ref{apx:incentive-property-of-rr-alg} for more discussion about the incentive issue.

\section{Contiguous Setting with Identical Valuations}
\label{sec:contiguous-setting-with-identical-valuations}

In this section, we assume that all items are arranged on a line by the order they arrive.
Moreover, we impose the constraint on the allocation that it has to be contiguous, i.e., each agent gets a contiguous block of items.
We assume valuations to be identical in this section, which can be represented by a valuation $v$ shared by all agents, and agent $i$ always gets the $i$-th block from left to right.
We will establish lower and upper bounds for both PROPa and EF1.
All the omitted proofs in this section can be found in Appendix~\ref{apx:missing-proofs-of-section-contiguous-setting-with-identical-valuations}.

\subsection{Upper Bounds for PROPa}

In the offline model, for any line of items $M$ and $n$ agents with nonidentical valuations, Suksompong \cite{DBLP:journals/dam/Suksompong19} designs an efficient algorithm to compute a PROPa contiguous allocation.
He also shows that $\frac{n - 1}{n} \cdot v^{\max}$ is the best additive approximation factor even for identical valuations and without the contiguity requirement.

\begin{lemma}[Theorem 1 in \cite{DBLP:journals/dam/Suksompong19}]\label{proportional-lemma}
    Suppose that there are $n$ agents with nonidentical valuations and the items in $M$ are arranged on a line.
    Then contiguous PROPa allocations always exist and can be computed in polynomial time.
\end{lemma}

We start by describing the offline algorithm given by Suksompong \cite{DBLP:journals/dam/Suksompong19}.
The algorithm maintains a current block, iterates all items sequentially, and adds the new item to the current block.
In the beginning, all agents are active.
Whenever there exists an agent $i$ satisfying PROPa for the current block, then the current block is obtained by agent $i$.
After that, agent $i$ is deactivated and the current block becomes empty.
Finally, when all agents become inactive, the remaining items are assigned arbitrarily without violating the contiguity constraints.

For identical valuations in the online setting, we give an algorithm based on the offline algorithm that achieves PROPa with $O(nT)$ adjustments in Algorithm~\ref{online-proportional-algorithm-for-additive-identical-agents}.
Intuitively, we intend to use the offline algorithm as a subroutine to compute an allocation in each round independently.
If we successfully prove that each separating point between two adjacent blocks is nondecreasing throughout all rounds, then the $O(nT)$ upper bound is established.
A sufficient condition for the nondecreasing property of the separating points is that the constraint factor in round $t$, defined as $B_t := \frac{1}{n} \cdot v(M_t) - \frac{n - 1}{n} \cdot v^{\max}(M_t)$, is nondecreasing.
Unfortunately, this may not hold.
For example, if $n = 3$ and there are two items $g_1, g_2$ such that $v(g_1) = 0$ and $v(g_2) = 1$, then we have $B_1 = 0$ and $B_2 = -1/3 < B_1$.
The essential observation is that when $B_{t + 1} < B_t$ and $A^t$ satisfies the constraint factor $B_t$, adding an item to some block will only enable the allocation to be feasible for a larger constraint factor.
Therefore, if $B_{t + 1} < B_t$, we just arbitrarily assign the new item to obtain $A^{t + 1}$ from $A^t$.
In fact, we will show that if $B_{t + 1} < B_t^{\max}$ where $B_t^{\max} := \max_{i \in [t]} B_i$, then we are allowed to allocate the new item arbitrarily.

In Algorithm~\ref{online-proportional-algorithm-for-additive-identical-agents}, for any $i \in [n]$, $p_i$ serves as the index of the last item in the $i$-th block, where $p_0 = 0$ and $p_n = t$ are omitted.
We update $p_i$ sequentially by moving $p_i$ backward until the $i$-th block satisfies the proportional constraint, i.e., $v(M_{p_{i - 1}, p_{i}}) \geq B_t$ (Line~\ref{algoline:move-pointer1}--Line~\ref{algoline:move_pointer2}).

\begin{algorithm}[!htp]
    \caption{PROPa Algorithm for Identical Valuations}
    \label{online-proportional-algorithm-for-additive-identical-agents}
    \begin{algorithmic}[1]
        \Require {A valuation function $v$ for all agents}
        \State {$p_i \leftarrow 0$ for $i = 1, 2, \ldots, n-1$}
        \For {$t = 1, \ldots, T$}
            \State {$B_t \leftarrow \frac{1}{n} \cdot v(M_t) - \frac{n-1}{n} \cdot v^{\max}(M_t)$}
            \For {$i = 1, \ldots, n - 1$}
                \While {$v\left(M_{p_{i - 1}, p_i}\right) < B_t$}\label{algoline:move-pointer1}
                    \State {$p_i \leftarrow p_i + 1$}\label{algoline:move_pointer2}
                \EndWhile
                \State {$A^t_i \leftarrow M_{p_{i - 1}, p_i}$}
            \EndFor
            \State {$A^t_n \leftarrow M_{p_{n - 1}, t}$}
        \EndFor
        \State {\Return $[A^1, A^2, \ldots, A^T]$}
    \end{algorithmic}
\end{algorithm}

\begin{restatable}{theorem}{thmPROPaalgorithmforidenticalvaluations}
\label{thm:PROPa-algorithm-for-identical-valuations}
    Algorithm~\ref{online-proportional-algorithm-for-additive-identical-agents} is a PROPa algorithm for identical valuations that requires $O(nT)$ adjustments.
\end{restatable}

\subsection{Upper Bounds for EF1}
In this section, we consider EF1 as the fairness criterion, which, by Proposition~\ref{prop:EF1-implies-PROPa}, is stronger than PROPa.
We first give an EF1 algorithm for two agents with general identical valuations that requires $O(T)$ adjustments in Algorithm~\ref{EF1-algorithm-for-two-agents}.
At each round $t$, we first find the minimum index $i$ such that the total valuation of items on the left of $g_i$ (inclusively) is at least the total valuation of the remaining items\footnote{Such $i$ is called \textit{lumpy tie} in \cite{DBLP:journals/geb/BiloCFIMPVZ22}, which possesses some desirable properties that are useful to establish contiguous fair allocations.}, i.e., $v(M_i) \geq v(M_{i, t})$ (Line~\ref{algoline:find-lumpy-tie}).
Then the line of items is separated into two blocks by $g_i$, where $g_i$ belongs to the block with the smaller valuation (Line~\ref{algoline:break-by-lumpy-tie1}--Line~\ref{algoline:break-by-lumpy-tie2}).

\begin{algorithm}[!htp]
        \caption{EF1 Algorithm for Two Agents with General Identical Valuations}
        \label{EF1-algorithm-for-two-agents}
        \begin{algorithmic}[1]
            \Require {A valuation function $v$ for all agents}
            \For {$t = 1, \ldots, T$}
                \State {$i = \min\{j \mid v(M_j) \geq v(M_{j, t})\}$}\label{algoline:find-lumpy-tie}
                \If {$v(M_{i - 1}) \leq v(M_{i, t})$}\label{algoline:break-by-lumpy-tie1}
                    \State {$A^t = (M_i, M_{i, t})$}
                    \Else \State {$A^t = (M_{i - 1}, M_{i - 1, t})$}\label{algoline:break-by-lumpy-tie2}
                \EndIf
            \EndFor
            \State {\Return $[A^1, A^2, \ldots, A^T]$}
        \end{algorithmic}
\end{algorithm}

\begin{theorem}
    Algorithm~\ref{EF1-algorithm-for-two-agents} is an EF1 algorithm for two agents with general identical valuations that requires $O(T)$ adjustments.
\end{theorem}

\begin{proof}
    First, the $i$ specified in Line~\ref{algoline:find-lumpy-tie} is well-defined since $t \in \{j \mid v(M_j) \geq v(M_{j, t})\}$.
    We show that $A^t$ is EF1 for any $t \in [T]$.
    Fix $t \in [T]$ and define $i$ as in Line~\ref{algoline:find-lumpy-tie}.
    Let $L = M_{i - 1}$ and $R = M_{i, t}$, and we have $v(L + g_i) \geq v(R)$ by definition.
    Note that $v(R + g_i) > v(L)$ holds since otherwise $v(L) \geq v(R + g_i)$ and we know that $i - 1 \in \{j \mid v(M_j) \geq v(M_{j, t})\}$, which contradicts the minimality of $i$.
    Without loss of generality, suppose that $A^t = (L \cup \{g_i\}, R)$ by symmetry, which means that $v(L) \leq v(R)$.
    Thus agent $2$ does not envy agent $1$ up to one item.
    On the other hand, $v(L + g_i) \geq v(R)$ and thus agent $1$ does not envy agent $2$.
    As a result, $A^t$ is EF1 for any $t \in [T]$.
    
    It remains to prove the required number of adjustments.
    Define $f(t) = \min\{j \mid v(M_j) \geq v(M_{j, t})\}$ and $h(t)$ such that agent $1$ obtains exactly the first $h(t)$ items in $A^t$.
    By the monotonicity of $v$, we know that $f(t)$ is nondecreasing.
    Due to the allocation rule, we have $f(t) = h(t)$ or $f(t) = h(t) - 1$.
    It suffices to prove that $h(t)$ is nondecreasing since in this case, the total number of adjustments is
    \begin{align*}
        \sum_{t=1}^{T - 1} |h(t + 1) - h(t)| = \sum_{t=1}^{T - 1} (h(t + 1) - h(t)) = h(T) - h(1) \leq T.
    \end{align*}
    Fix $t < T$.
    If $f(t) < f(t + 1)$, it holds that
    \begin{align*}
        h(t) \leq f(t) \leq f(t + 1) - 1 \leq h(t + 1).
    \end{align*}
    Besides, if $f(t) = f(t + 1)$ and $h(t) = f(t) - 1$, then $h(t) \leq h(t + 1)$ trivially holds.
    Lastly, if $f(t) = f(t + 1)$ and $h(t) = f(t)$, we have $v(M_{f(t) - 1}) \leq v(M_{f(t), t})$.
    Since 
    \begin{align*}
        v(M_{f(t + 1) - 1})
        & = v(M_{f(t) - 1})
        \leq v(M_{f(t), t})\\
        & \leq v(M_{f(t), t + 1})
        = v(M_{f(t + 1), t + 1}),
    \end{align*}
    $g_{f(t + 1)}$ is allocated to agent $1$ in $A^{t + 1}$.
    It follows that $h(t + 1) = f(t + 1) = h(t)$.
    Therefore, $h(t)$ is nondecreasing.
\end{proof}

For any number of agents, we fail to establish such a nondecreasing property for the separating points, by which we can similarly achieve $O(nT)$ adjustments.
Nevertheless, under the assumption that the valuation of each item lies in some range $[L, R]$ such that $0 < L \leq R$, we employ the slightly modified version of Algorithm 1 in \cite{DBLP:journals/geb/BiloCFIMPVZ22} for any number of agents with additive identical valuations as a subroutine to obtain an EF1 allocation in each round independently.
This leads to an EF1 algorithm for additive identical valuations that requires $O\left( (R / L) \cdot n^2T\right)$ adjustments.

We present the offline algorithm in Algorithm~\ref{alg:EF1-offline-algorithm-for-arbitrary-many-agents} and the online algorithm in Algorithm~\ref{alg:EF1-online-algorithm-for-arbitrary-many-agents-with-range-assumption}.
For a contiguous allocation $A$, define $P(A) = (\ell_0, \ldots, \ell_n)$ such that $\ell_0 = 0$ and $P_i(A) = \ell_i$ is the index of the last item in $A_i$.
Denote the \textit{leximin (lexicographic maximin)} allocation as the allocation that maximizes the lowest valuation among all agents, subject to which it maximizes the second-lowest one, and so on.
We say that $A$ is leximinmin if $A$ is the leximin allocation with lexicographically minimum $P(A)$.
In Algorithm~\ref{alg:EF1-offline-algorithm-for-arbitrary-many-agents}, we first find a contiguous leximinmin allocation\footnote{For completeness, we show how to use Dynamic Programming to efficiently compute the contiguous leximinmin allocation in Appendix~\ref{apx:dynamic-programming-to-compute-leximin2-allocations}.} $A$ as the initial allocation and fix an agent $i$ with the minimum valuation in $A$.
Then we iterate from $1$ to $i - 1$ and from $n$ to $i + 1$ separately.
During each iteration $j$, whenever agent $i$ envies agent $j$ even up to one item, the item in $A_j$ that is closest to $A_i$ is moved to the adjacent block.
The only difference between Algorithm~\ref{alg:EF1-offline-algorithm-for-arbitrary-many-agents} in this paper and Algorithm 1 in \cite{DBLP:journals/geb/BiloCFIMPVZ22} is that Algorithm~\ref{alg:EF1-offline-algorithm-for-arbitrary-many-agents} uses the leximinmin allocation as the initial allocation rather than an arbitrary leximin allocation.
Intuitively, choosing the leximinmin allocation guarantees the property that each separating point between two adjacent blocks is nondecreasing (Lemma~\ref{lem:nondecreasing-leximin-allocations}), which is crucial for proving the upper bound.

\begin{algorithm}[!htp]
    \caption{EF1 Offline Algorithm}
    \label{alg:EF1-offline-algorithm-for-arbitrary-many-agents}
    \begin{algorithmic}[1]
        \Require {A valuation function $v$ for all agents and an item set $M$}
        \State {Let $A = (A_1, \ldots, A_n)$ be the contiguous leximinmin allocation of $M$}
        \State {Fix $i = \arg\min_{j \in N} v(A_j)$}\label{algoline:choosei}
        \For {$j = 1, \ldots, i - 1$}\label{algoline:start-loop}
            \While {agent $i$ envies agent $j$ even up to one item}
                \State {move the rightmost item of $A_j$ to $A_{j + 1}$}
            \EndWhile
        \EndFor
        \For {$j = n, \ldots, i + 1$}
            \While {agent $i$ envies agent $j$ even up to one item}
                \State {move the leftmost item of $A_j$ to $A_{j - 1}$}\label{algoline:end-loop}
            \EndWhile
        \EndFor
        \State {\Return $A$}
    \end{algorithmic}
\end{algorithm}

\begin{algorithm}[!htp]
    \caption{EF1 Online Algorithm}
    \label{alg:EF1-online-algorithm-for-arbitrary-many-agents-with-range-assumption}
    \begin{algorithmic}[1]
        \Require {A valuation function $v$ for all agents}
        \For {$t = 1, \ldots, T$}
            \State {$A^t \leftarrow $ the allocation returned by Algorithm~\ref{alg:EF1-offline-algorithm-for-arbitrary-many-agents} running with $v$ and $M_t$}
        \EndFor
        \State {\Return $[A^1, A^2, \ldots, A^T]$}
    \end{algorithmic}
\end{algorithm}

We first present the property of Algorithm~\ref{alg:EF1-offline-algorithm-for-arbitrary-many-agents}.

\begin{lemma}[Theorem 7.1 in \cite{DBLP:journals/geb/BiloCFIMPVZ22}]\label{lem:property-of-offline-EF1-algorithm-for-general-identical-valuations}
    For additive identical valuations, Algorithm~\ref{alg:EF1-offline-algorithm-for-arbitrary-many-agents} is EF1 and can be implemented in polynomial time.
    Moreover, throughout the two loops in Algorithm~\ref{alg:EF1-offline-algorithm-for-arbitrary-many-agents} (Line~\ref{algoline:start-loop}--Line~\ref{algoline:end-loop}), $A_i$ does not change, where $i$ is determined in Line~\ref{algoline:choosei} of Algorithm~\ref{alg:EF1-offline-algorithm-for-arbitrary-many-agents}.
\end{lemma}

Now we give the performance of Algorithm~\ref{alg:EF1-online-algorithm-for-arbitrary-many-agents-with-range-assumption}.

\begin{theorem}\label{thm:EF1-algorithm-for-additive-identical-valuations}
    Suppose that the valuations are additive identical and there exist $0 < L \leq R$ such that $v(g) \in [L, R]$ for any $g \in M$.
    Then Algorithm~\ref{alg:EF1-online-algorithm-for-arbitrary-many-agents-with-range-assumption} is EF1 and requires $O\left((R / L) \cdot n^2T\right)$ adjustments.
\end{theorem}

Before proving Theorem~\ref{thm:EF1-algorithm-for-additive-identical-valuations}, we give some useful lemmas.
Observe that without the contiguity constraint, leximin allocations are EF1 for additive identical valuations.
Since otherwise by selecting agents $i, j$ such that agent $i$ envies agent $j$ even up to one item and moving an item with a positive valuation of agent $j$ to agent $i$, we reach a contradiction.
Nevertheless, it is possible that none of the contiguous leximin allocations is EF1.
For instance, suppose that there are $3$ agents and $5$ items with valuations 1, 3, 1, 1, 1.
In this instance, the only contiguous leximin allocation is $(\{g_1\}, \{g_2\}, \{g_3, g_4, g_5\})$ and agent $1$ envies agent $3$ even up to one item.

In a contiguous leximin allocation $A = (A_1, \ldots, A_n)$, an essential observation is that the difference between the valuations of two adjacent blocks is upper bounded by $R$, where $R$ is the maximum valuation of items.
As a result, let $i = \arg\min_{j \in N} v(A_j)$ be the agent with the minimum valuation, then $v(A_j) - v(A_i)$ is upper bounded by $|i - j| R$.
We show in the following lemma that such inequality always holds throughout Algorithm~\ref{alg:EF1-offline-algorithm-for-arbitrary-many-agents}

\begin{restatable}{lemma}{lemleximinallocationswithboundedvaluations}
\label{lem:leximin-allocations-with-bounded-valuations}
    For additive identical valuations, suppose that the valuation of each item is at most $R$.
    Let $i = \arg\min_{j \in N} v(A_j)$ and $B = v(A_i)$ where $A$ is the initial contiguous leximinmin allocation.
    Then throughout Algorithm~\ref{alg:EF1-offline-algorithm-for-arbitrary-many-agents}, $v(A_j) \leq |i - j| R + B$ always holds for any $j \in N$.
\end{restatable}

\begin{proof}
    Due to symmetry, we assume that $i = 1$ and we will prove that
    \begin{align}
        v(A_j) \leq (j - 1) R + B \label{eqn:bounded-valuation-for-leximin-allocations}
    \end{align}
    always holds for any $j \in N$.
    Note that by Lemma~\ref{lem:property-of-offline-EF1-algorithm-for-general-identical-valuations}, $A_i$ does not change throughout the algorithm.
    Thus \eqref{eqn:bounded-valuation-for-leximin-allocations} always holds for $j = 1$.
    Now supposing for induction that \eqref{eqn:bounded-valuation-for-leximin-allocations} always holds for $j = k - 1$ where $k > 1$, we will show that \eqref{eqn:bounded-valuation-for-leximin-allocations} holds for $j = k$ all the time.
    Suppose for contradiction that $v(A_k) > (k - 1) R + B$ at some time.
    During iteration $k$, we will keep moving the leftmost item in $A_k$ to $A_{k - 1}$ until agent $i$ does not envy agent $k$ up to one item.
    Note that agent $i$ does not envy agent $k$ up to one item only if $v(A_k) \leq v(A_i) + R = B + R$, since otherwise for any $g \in A_k$, $v(A_k - g) \geq v(A_k) - R > v(A_i)$.
    Thus the total valuation of the items being moved from $A_k$ to $A_{k - 1}$ during iteration $k$ is at least
    \begin{align*}
        v(A_k) - (v(A_i) + R) > (k - 1) R + B - (B + R) = (k - 2) R,
    \end{align*}
    which, combining the fact that $v(A_{k - 1}) \geq B$, contradicts the induction assumption that \eqref{eqn:bounded-valuation-for-leximin-allocations} always holds for $j = k - 1$.
\end{proof}

The following lemma describes the monotonicity of the separating points between two adjacent blocks in the initial contiguous leximinmin allocations in Algorithm~\ref{alg:EF1-offline-algorithm-for-arbitrary-many-agents}.
This is essentially why we choose leximinmin allocations as initial allocations.
Recall that $P_i(A)$ is the index of the last item in bundle $A_i$.

\begin{lemma}\label{lem:nondecreasing-leximin-allocations}
    For additive identical valuations, let $X^t$ be the initial contiguous leximinmin allocation fixed by Algorithm~\ref{alg:EF1-offline-algorithm-for-arbitrary-many-agents} in round $t$.
    Then for every $t < T$ and $j \in N$, we have $P_j(X^t) \leq P_j(X^{t + 1})$.
\end{lemma}

Now it is sufficient to prove Theorem~\ref{thm:EF1-algorithm-for-additive-identical-valuations}.

\begin{proof}[Proof of Theorem~\ref{thm:EF1-algorithm-for-additive-identical-valuations}]
    Lemma~\ref{lem:property-of-offline-EF1-algorithm-for-general-identical-valuations} implies that Algorithm~\ref{alg:EF1-online-algorithm-for-arbitrary-many-agents-with-range-assumption} is EF1.
    It remains to prove the required number of adjustments.
    We will use Lemma~\ref{lem:leximin-allocations-with-bounded-valuations} to show that for any $t \in [T]$, the number of adjustments required to transform $X^t$ to $A^t$ is upper bounded by $O(n^2R / L)$.
    Combining the monotonicity of $P_j(X^t)$ given by Lemma~\ref{lem:nondecreasing-leximin-allocations}, we can establish the desired upper bound by triangle inequality.
    
    Fix round $t$ and let $i = \arg\min_{j \in N} v(X^t_j)$.
    By Lemma~\ref{lem:leximin-allocations-with-bounded-valuations}, we know that $v(A_j) \leq v(A_i) + |i - j| R$ throughout Algorithm~\ref{alg:EF1-offline-algorithm-for-arbitrary-many-agents}.
    Since by Lemma~\ref{lem:property-of-offline-EF1-algorithm-for-general-identical-valuations}, $A_i$ never changes and the valuation of each item is at least $L$, the number of items being moved from $A_j$ toward $A_i$ during the iteration $j$ is at most $(v(A_j) - v(A_i)) / L \leq |i - j| R / L$.
    Note that for two contiguous allocations $A'$ and $A''$, the number of adjustments that are required to transform $A'$ to $A''$ is $\sum_{j=1}^{n - 1} |P_j(A') - P_j(A'')|$.
    Thus the number of adjustments required by Algorithm~\ref{alg:EF1-offline-algorithm-for-arbitrary-many-agents} to transform the initial contiguous leximinmin allocation $X^t$ to the final EF1 allocation $A^t$ is
    \begin{align*}
        \sum_{j=1}^{n - 1} |P_j(X^t) - P_j(A^t)|
        \leq \sum_{j \neq i} |i - j|R / L \leq O(n^2R / L).
    \end{align*}
    Therefore, the number of adjustments required by Algorithm~\ref{alg:EF1-online-algorithm-for-arbitrary-many-agents-with-range-assumption} is
    \begin{align*}
        &\sum_{t=1}^{T - 1} \sum_{j=1}^{n - 1} |P_j(A^t) - P_j(A^{t + 1})|\\
        \leq & \sum_{t=1}^{T - 1} \sum_{j=1}^{n - 1} \Big(|P_j(A^t) - P_j(X^t)| + |P_j(X^t) - P_j(X^{t + 1})| \\
        & + |P_j(X^{t + 1}) - P_j(A^{t + 1})|\Big)\\
        \leq & \sum_{t=1}^{T - 1} \sum_{j=1}^{n - 1} |P_j(X^t) - P_j(X^{t + 1})| + O\left((R / L) \cdot n^2T\right)\\
        = & \sum_{t=1}^{T - 1} \sum_{j=1}^{n - 1} (P_j(X^{t + 1}) - P_j(X^t)) + O\left((R / L) \cdot n^2T\right)\\
        = & \sum_{j=1}^{n - 1}(P_j(X^T) - P_j(X^1)) + O\left((R / L) \cdot n^2T\right)\\
        \leq &O(nT) + O\left((R / L) \cdot n^2T\right) = O\left((R / L) \cdot n^2T\right),
    \end{align*}
    where the first equality is due to the monotonicity of $P_j(X^t)$ for any $j \in N$ given by Lemma~\ref{lem:nondecreasing-leximin-allocations}.
\end{proof}

\subsection{Lower Bounds}

We demonstrate that our upper bound to achieve PROPa for identical valuations is asymptotically tight.
Moreover, by Proposition~\ref{prop:EF1-implies-PROPa}, the same lower bound is also applied to EF1.

\begin{theorem}\label{thm:lower-bound-for-identical-valuations-and-proportionality}
    For identical valuations, every PROPa algorithm requires at least $\Omega(nT)$ adjustments.
\end{theorem}

\begin{proof}
    Suppose that there are $n$ agents with identical valuations and each agent has a valuation of $1$ for every item.
    The proportional constraint factor for $M_T$ that should be satisfied by each block is
    \begin{align*}
        \left \lceil \frac{T}{n} - \frac{n-1}{n} \right \rceil = \left\lfloor \frac{T}{n} \right\rfloor.
    \end{align*}
    Denote the number of items that have arrived as $T$.
    We first prove that, for a large enough $t$, $g_t$ must come to every agent's block at least once as $T$ grows larger.
    In other words, for every agent $i$, there exists $T_i$ such that in round $T_i$, $g_t$ belongs to agent $i$.
    Then we use this property to show the desired lower bound.
    
    We will prove that, for $t \geq n^2$, $g_t$ must come to every agent's block at least once as $T$ grows larger.
    Fix $t \geq n^2$.
    Note that in round $T = t$, $g_t$ must belong to agent $n$, since agent $n$'s block cannot be empty.
    Now we show that for any $i < n$, agent $i$ will obtain $g_t$ at some time.
    Suppose $t = ki - r$, where $0 \leq r < i$ and $k \geq n$.
    In round $T = kn > t$, each agent should obtain exactly $k$ items.
    Then $g_t$ belongs to agent $i$ since $k (i - 1) < t \leq k i$.
    
    Now it is sufficient to give the lower bound.
    By the above arguments, for any $t \geq n^2$, $T \geq t$ and $i \geq nt / (T - n)$, we know that $g_t$ must have been obtained by agents $i, i + 1, \ldots, n$ previously.
    Thus the number of adjustments made on $g_t$ throughout $T$ rounds is at least $n - \lceil nt / (T - n) \rceil$.
    Summing over all $t \geq n^2$, we have
    \begin{align*}
        \sum_{t = n^2}^T \left( n - \left \lceil \frac{nt}{T - n} \right \rceil \right)
        &\geq \sum_{t = n^2}^T \left( n - \frac{nt}{T - n} - 1 \right)\\
        &\geq (n - 1) (T - n^2)  - \sum_{t=1}^T \frac{nt}{T - n}\\
        &= (n - 1) (T - n^2) - \frac{nT(T + 1)}{2 (T - n)}\\
        &= \Omega(nT),
    \end{align*}
    where the last equality is due to $T \gg n$.
    Therefore, the number of adjustments required by any PROPa algorithm is at least $\Omega(nT)$.
\end{proof}

\begin{corollary}
    For identical valuations, every EF1 algorithm requires at least $\Omega(nT)$ adjustments.
\end{corollary}

\section{Contiguous Setting with Nonidentical Valuations}
\label{sec:contiguous-setting-with-nonidentical-valuations}

In this section, for nonidentical valuations with contiguity requirement, we show that it is hopeless to make any significant improvement to the trivial $O(T^2)$ upper bound for both PROPa and EF1.
All the omitted proofs in this section can be found in Appendix~\ref{apx:missing-proofs-of-contiguous-nonidentical-setting}.

\begin{theorem}\label{thm:non-identical-lower-bound-for-proportionality}
    For nonidentical valuations, every PROPa algorithm requires at least $\Omega(T^2 / n)$ adjustments.
\end{theorem}

\begin{proof}
    \begin{figure}[t]
        \centering


\begin{tikzpicture}[thick,scale=0.6, every node/.style={transform shape}]
    \node[font=\fontsize{15}{6}\selectfont] at (0,0) {$v_n(g_t)$};
    \node[font=\fontsize{15}{6}\selectfont] at (0,1) {$\vdots$};
    \node[font=\fontsize{15}{6}\selectfont] at (0,2) {$v_2(g_t)$};
    \node[font=\fontsize{15}{6}\selectfont] at (0,3) {$v_1(g_t)$};
    
    \node[font=\fontsize{15}{6}\selectfont] at (1,3) {$1$};
    \node[font=\fontsize{15}{6}\selectfont] at (1.5,3) {$\ldots$};
    \node[font=\fontsize{15}{6}\selectfont] at (2,3) {$1$};
    
    \node[font=\fontsize{15}{6}\selectfont] at (2.5,2) {$1$};
    \node[font=\fontsize{15}{6}\selectfont] at (3,2) {$\ldots$};
    \node[font=\fontsize{15}{6}\selectfont] at (3.5,2) {$1$};
    
    \node[font=\fontsize{15}{6}\selectfont] at (4,1) {$\ldots$};
    
    \node[font=\fontsize{15}{6}\selectfont] at (4.5,0) {$1$};
    \node[font=\fontsize{15}{6}\selectfont] at (5,0) {$\ldots$};
    \node[font=\fontsize{15}{6}\selectfont] at (5.5,0) {$1$};
    
    \node[font=\fontsize{15}{6}\selectfont] at (6.2,3) {$n^2$};
    \node[font=\fontsize{15}{6}\selectfont] at (6.9,3) {$\ldots$};
    \node[font=\fontsize{15}{6}\selectfont] at (7.6,3) {$n^2$};
    
    \node[font=\fontsize{15}{6}\selectfont] at (8.2,2) {$n^2$};
    \node[font=\fontsize{15}{6}\selectfont] at (8.9,2) {$\ldots$};
    \node[font=\fontsize{15}{6}\selectfont] at (9.6,2) {$n^2$};
    
    \node[font=\fontsize{15}{6}\selectfont] at (10.2,1) {$\ldots$};
    
    \node[font=\fontsize{15}{6}\selectfont] at (10.9,0) {$n^2$};
    \node[font=\fontsize{15}{6}\selectfont] at (11.6,0) {$\ldots$};
    \node[font=\fontsize{15}{6}\selectfont] at (12.3,0) {$n^2$};
    
    \node[font=\fontsize{15}{6}\selectfont] at (13.2,3) {$\ldots$};
    \node[font=\fontsize{15}{6}\selectfont] at (13.2,2) {$\ldots$};
    \node[font=\fontsize{15}{6}\selectfont] at (13.2,1) {$\ldots$};
    \node[font=\fontsize{15}{6}\selectfont] at (13.2,0) {$\ldots$};
    
    \draw [decorate,decoration = {brace,amplitude=2mm},ultra thick] (0.9, 3.3) -- (5.7, 3.3);
    \node[font=\fontsize{15}{6}\selectfont] at (3.4,4) {$n^2$};
    
    \draw [decorate,decoration = {brace,mirror,amplitude=2mm},ultra thick] (4.4, -0.3) -- (5.6, -0.3);
    \node[font=\fontsize{15}{6}\selectfont] at (5,-0.9) {$n$};
\end{tikzpicture}
        \caption{Figure illustrating the instance in the proof of Theorem~\ref{thm:non-identical-lower-bound-for-proportionality}.
        \label{fig:lower-bound-of-non-identical-valuations-for-prop}
        Each period of length $n^2$ is divided into $n$ blocks of length $n$, and agent $i$ is only interested in the items in the $i$-th block in each period $c$ with $n^{2c}$ valuation for each of these items.
        The valuation $0$ is omitted.}
    \end{figure}

    Suppose that there are $T$ items and $n$ agents with valuations $v_1, \ldots, v_n$.
    For any $t \in [T]$, let
    \begin{align*}
        v_i(g_t) = 
        \begin{cases}
            n^{2c}, & cn^2 + (i - 1) n < t \leq cn^2 + in \text{ for some } c \geq 0,\\
            0, & \text{otherwise},
        \end{cases}
    \end{align*}
    for all $i \in N$.
    As illustrated in Figure~\ref{fig:lower-bound-of-non-identical-valuations-for-prop}, each period of length $n^2$ is divided into $n$ blocks of length $n$, and agent $i$ is only interested in the items in the $i$-th block in each period $c$ with $n^{2c}$ valuation for each of these items.
    It suffices to show that, for every $k \geq 2n$, the first $nk - n^2$ items belong to one agent in round $nk$ and another agent in round $nk + n$.
    Since it follows that the number of adjustments is at least $\sum_{k=2n}^{T/n - 1} (nk - n^2) = \Theta(T^2 / n)$, where we assume $T \gg n$.
     
    We only give the proof for $k$ such that $k$ is a multiple of $n$, and the proof can be easily generalized to any $k \geq 2n$.
    Now we show that for any $c \geq 2$ and $k = cn$, the first $nk - n^2 = (c - 1) n^2$ items belong to agent $1$ in round $nk$ and belong to agent $2$ in round $nk + n$.
    Note that in round $nk$, the proportional constraint factor for all agents is
    \begin{align*}
        & \frac{1}{n} \left( n \sum_{j=0}^{c - 1} n^{2j} \right) - \frac{n - 1}{n} \cdot n^{2(c - 1)}
        = n^{2c - 3} + \sum_{j=0}^{c - 2} n^{2j},
    \end{align*}
    and the total valuation of the first $nk - n^2$ items for each agent is
    \begin{align*}
        n \sum_{j=0}^{c - 2} n^{2j} = \frac{n^{2(c - 1)} - 1}{n + 1} + \sum_{j=0}^{c - 2} n^{2j} < n^{2c - 3} + \sum_{j=0}^{c - 2} n^{2j}.
    \end{align*}
    Define $G_i = \{g_t \mid (c - 1) n^2 + (i - 1) n < t \leq (c - 1) n^2 + in\}$ as the set of items that arrive during period $c$ and agent $i$ is interested in.
    To satisfy the proportional constraint, each agent $i$ should get at least one of the items in $G_i$.
    Due to the contiguity requirement, for every $i \in N$, agent $i$ must get the $i$-th block which should contain at least one item in $G_i$.
    As a result, the first $nk - n^2 = (c - 1) n^2$ items belong to agent $1$.
    Similarly, we can show that in round $nk + n$, the first $nk - n^2 + n$ items belong to agent $2$ and we are done.
\end{proof}

When it comes to EF1, the lower bound is even stronger, indicating that we cannot make any improvement to the trivial $O(T^2)$ upper bound.
The proof of Theorem~\ref{thm:non-identical-lower-bound-for-EF1} is analogous to the proof of Theorem~\ref{thm:non-identical-lower-bound-for-proportionality} except that the length of a period is reduced from $n^2$ to $O(n)$.

\begin{restatable}{theorem}{thmnonidenticallowerboundforEF}
\label{thm:non-identical-lower-bound-for-EF1}
    For nonidentical valuations, every EF1 algorithm requires at least $\Omega(T^2)$ adjustments.
\end{restatable}

Note that in the hard instance given to prove Theorem~\ref{thm:non-identical-lower-bound-for-proportionality}, the valuations of items are unbounded.
Nevertheless, even for $2$ agents with binary valuations, i.e. $v_i(g) \in \{0, 1\}$ for all $i$ and $g$, the $\Omega(T^2)$ lower bound still pertains\footnote{For $2$ agents with binary valuations, PROPa is equivalent to EF1.}.

\begin{restatable}{theorem}{thmlowerboundforbinaryagents}
\label{thm:lower-bound-for-2-binary-agents}
    For $2$ agents with binary valuations, every EF1 algorithm requires at least $\Omega(T^2)$ adjustments.
\end{restatable}
\section{Conclusion and Future Work}

We conclude with some directions for future work.
\begin{itemize}
    \item Even though we have established almost tight upper and lower bounds to achieve PROPa in the contiguous setting or when $n = 2$, there are still large gaps between the upper and lower bounds to achieve EF1 in both continuous and noncontinuous settings.
    The first future direction is to tighten these bounds.
    
    \item If the types of items are drawn from certain distributions rather than chosen adversely, can we show a better upper bound in expectation or asymptotically in both settings\footnote{Similar problems are presented in \cite{DBLP:journals/dam/Suksompong19}.}?
    
    \item It would be interesting to investigate other fairness notions like EQ1 \cite{DBLP:journals/corr/abs-2101-09794}.
    
    
    \item Another promising direction, like always asked in  the offline setting \cite{DBLP:journals/corr/abs-2204-14229}, is to achieve fairness and Pareto optimality simultaneously.
    This has been shown in some other online models \cite{DBLP:conf/sigecom/ZengP20, DBLP:conf/pricai/AleksandrovW19}.
\end{itemize}
\begin{acks}
	We would like to thank Alexandros Psomas for the discussion about the content in the early stage of this work and his helpful suggestion on the presentation.
    We are also grateful to Ruta Mehta for inspiring us to think about the incentive issue in this model and to anonymous reviewers for their useful comments.
\end{acks}



\bibliographystyle{ACM-Reference-Format} 
\bibliography{reference}


\begin{thebibliography}{27}


\ifx \showCODEN    \undefined \def \showCODEN     #1{\unskip}     \fi
\ifx \showDOI      \undefined \def \showDOI       #1{#1}\fi
\ifx \showISBNx    \undefined \def \showISBNx     #1{\unskip}     \fi
\ifx \showISBNxiii \undefined \def \showISBNxiii  #1{\unskip}     \fi
\ifx \showISSN     \undefined \def \showISSN      #1{\unskip}     \fi
\ifx \showLCCN     \undefined \def \showLCCN      #1{\unskip}     \fi
\ifx \shownote     \undefined \def \shownote      #1{#1}          \fi
\ifx \showarticletitle \undefined \def \showarticletitle #1{#1}   \fi
\ifx \showURL      \undefined \def \showURL       {\relax}        \fi
\providecommand\bibfield[2]{#2}
\providecommand\bibinfo[2]{#2}
\providecommand\natexlab[1]{#1}
\providecommand\showeprint[2][]{arXiv:#2}

\bibitem[\protect\citeauthoryear{Akrami, Rezvan, and Seddighin}{Akrami
  et~al\mbox{.}}{2022}]%
        {DBLP:conf/ijcai/AkramiRS22}
\bibfield{author}{\bibinfo{person}{Hannaneh Akrami}, \bibinfo{person}{Rojin
  Rezvan}, {and} \bibinfo{person}{Masoud Seddighin}.}
  \bibinfo{year}{2022}\natexlab{}.
\newblock \showarticletitle{An {EF2X} Allocation Protocol for Restricted
  Additive Valuations}. In \bibinfo{booktitle}{\emph{{IJCAI}}}.
  \bibinfo{publisher}{ijcai.org}, \bibinfo{pages}{17--23}.
\newblock


\bibitem[\protect\citeauthoryear{Aleksandrov and Walsh}{Aleksandrov and
  Walsh}{2017}]%
        {DBLP:conf/ki/AleksandrovW17}
\bibfield{author}{\bibinfo{person}{Martin Aleksandrov} {and}
  \bibinfo{person}{Toby Walsh}.} \bibinfo{year}{2017}\natexlab{}.
\newblock \showarticletitle{Expected Outcomes and Manipulations in Online Fair
  Division}. In \bibinfo{booktitle}{\emph{{KI}}}
  \emph{(\bibinfo{series}{Lecture Notes in Computer Science},
  Vol.~\bibinfo{volume}{10505})}. \bibinfo{publisher}{Springer},
  \bibinfo{pages}{29--43}.
\newblock


\bibitem[\protect\citeauthoryear{Aleksandrov and Walsh}{Aleksandrov and
  Walsh}{2019}]%
        {DBLP:conf/pricai/AleksandrovW19}
\bibfield{author}{\bibinfo{person}{Martin Aleksandrov} {and}
  \bibinfo{person}{Toby Walsh}.} \bibinfo{year}{2019}\natexlab{}.
\newblock \showarticletitle{Strategy-Proofness, Envy-Freeness and Pareto
  Efficiency in Online Fair Division with Additive Utilities}. In
  \bibinfo{booktitle}{\emph{{PRICAI} {(1)}}} \emph{(\bibinfo{series}{Lecture
  Notes in Computer Science}, Vol.~\bibinfo{volume}{11670})}.
  \bibinfo{publisher}{Springer}, \bibinfo{pages}{527--541}.
\newblock


\bibitem[\protect\citeauthoryear{Aleksandrov and Walsh}{Aleksandrov and
  Walsh}{2020}]%
        {DBLP:conf/aaai/AleksandrovW20}
\bibfield{author}{\bibinfo{person}{Martin Aleksandrov} {and}
  \bibinfo{person}{Toby Walsh}.} \bibinfo{year}{2020}\natexlab{}.
\newblock \showarticletitle{Online Fair Division: {A} Survey}. In
  \bibinfo{booktitle}{\emph{{AAAI}}}. \bibinfo{publisher}{{AAAI} Press},
  \bibinfo{pages}{13557--13562}.
\newblock


\bibitem[\protect\citeauthoryear{Bai and G{\"{o}}lz}{Bai and
  G{\"{o}}lz}{2022}]%
        {DBLP:conf/ijcai/BaiG22}
\bibfield{author}{\bibinfo{person}{Yushi Bai} {and} \bibinfo{person}{Paul
  G{\"{o}}lz}.} \bibinfo{year}{2022}\natexlab{}.
\newblock \showarticletitle{Envy-Free and Pareto-Optimal Allocations for Agents
  with Asymmetric Random Valuations}. In \bibinfo{booktitle}{\emph{{IJCAI}}}.
  \bibinfo{publisher}{ijcai.org}, \bibinfo{pages}{53--59}.
\newblock


\bibitem[\protect\citeauthoryear{Benad{\`e}, Halpern, and Psomas}{Benad{\`e}
  et~al\mbox{.}}{[n.d.]}]%
        {benadedynamic}
\bibfield{author}{\bibinfo{person}{Gerdus Benad{\`e}}, \bibinfo{person}{Daniel
  Halpern}, {and} \bibinfo{person}{Alexandros Psomas}.}
  \bibinfo{year}{[n.d.]}\natexlab{}.
\newblock \showarticletitle{Dynamic Fair Division with Partial Information}.
\newblock  (\bibinfo{year}{[n.\,d.]}).
\newblock


\bibitem[\protect\citeauthoryear{Benade, Kazachkov, Procaccia, and
  Psomas}{Benade et~al\mbox{.}}{2018}]%
        {DBLP:conf/sigecom/BenadeKPP18}
\bibfield{author}{\bibinfo{person}{Gerdus Benade},
  \bibinfo{person}{Aleksandr~M. Kazachkov}, \bibinfo{person}{Ariel~D.
  Procaccia}, {and} \bibinfo{person}{Christos{-}Alexandros Psomas}.}
  \bibinfo{year}{2018}\natexlab{}.
\newblock \showarticletitle{How to Make Envy Vanish Over Time}. In
  \bibinfo{booktitle}{\emph{{EC}}}. \bibinfo{publisher}{{ACM}},
  \bibinfo{pages}{593--610}.
\newblock


\bibitem[\protect\citeauthoryear{Bil{\`{o}}, Caragiannis, Flammini, Igarashi,
  Monaco, Peters, Vinci, and Zwicker}{Bil{\`{o}} et~al\mbox{.}}{2022}]%
        {DBLP:journals/geb/BiloCFIMPVZ22}
\bibfield{author}{\bibinfo{person}{Vittorio Bil{\`{o}}},
  \bibinfo{person}{Ioannis Caragiannis}, \bibinfo{person}{Michele Flammini},
  \bibinfo{person}{Ayumi Igarashi}, \bibinfo{person}{Gianpiero Monaco},
  \bibinfo{person}{Dominik Peters}, \bibinfo{person}{Cosimo Vinci}, {and}
  \bibinfo{person}{William~S. Zwicker}.} \bibinfo{year}{2022}\natexlab{}.
\newblock \showarticletitle{Almost envy-free allocations with connected
  bundles}.
\newblock \bibinfo{journal}{\emph{Games Econ. Behav.}}  \bibinfo{volume}{131}
  (\bibinfo{year}{2022}), \bibinfo{pages}{197--221}.
\newblock


\bibitem[\protect\citeauthoryear{Bouveret, Cechl{\'{a}}rov{\'{a}}, Elkind,
  Igarashi, and Peters}{Bouveret et~al\mbox{.}}{2017}]%
        {DBLP:conf/ijcai/BouveretCEIP17}
\bibfield{author}{\bibinfo{person}{Sylvain Bouveret},
  \bibinfo{person}{Katar{\'{\i}}na Cechl{\'{a}}rov{\'{a}}},
  \bibinfo{person}{Edith Elkind}, \bibinfo{person}{Ayumi Igarashi}, {and}
  \bibinfo{person}{Dominik Peters}.} \bibinfo{year}{2017}\natexlab{}.
\newblock \showarticletitle{Fair Division of a Graph}. In
  \bibinfo{booktitle}{\emph{{IJCAI}}}. \bibinfo{publisher}{ijcai.org},
  \bibinfo{pages}{135--141}.
\newblock


\bibitem[\protect\citeauthoryear{Bouveret, Cechl{\'{a}}rov{\'{a}}, and
  Lesca}{Bouveret et~al\mbox{.}}{2019}]%
        {DBLP:journals/aamas/BouveretCL19}
\bibfield{author}{\bibinfo{person}{Sylvain Bouveret},
  \bibinfo{person}{Katar{\'{\i}}na Cechl{\'{a}}rov{\'{a}}}, {and}
  \bibinfo{person}{Julien Lesca}.} \bibinfo{year}{2019}\natexlab{}.
\newblock \showarticletitle{Chore division on a graph}.
\newblock \bibinfo{journal}{\emph{Auton. Agents Multi Agent Syst.}}
  \bibinfo{volume}{33}, \bibinfo{number}{5} (\bibinfo{year}{2019}),
  \bibinfo{pages}{540--563}.
\newblock


\bibitem[\protect\citeauthoryear{Brustle, Dippel, Narayan, Suzuki, and
  Vetta}{Brustle et~al\mbox{.}}{2020}]%
        {DBLP:conf/sigecom/BrustleDNSV20}
\bibfield{author}{\bibinfo{person}{Johannes Brustle}, \bibinfo{person}{Jack
  Dippel}, \bibinfo{person}{Vishnu~V. Narayan}, \bibinfo{person}{Mashbat
  Suzuki}, {and} \bibinfo{person}{Adrian Vetta}.}
  \bibinfo{year}{2020}\natexlab{}.
\newblock \showarticletitle{One Dollar Each Eliminates Envy}. In
  \bibinfo{booktitle}{\emph{{EC}}}. \bibinfo{publisher}{{ACM}},
  \bibinfo{pages}{23--39}.
\newblock


\bibitem[\protect\citeauthoryear{Budish}{Budish}{2010}]%
        {DBLP:conf/bqgt/Budish10}
\bibfield{author}{\bibinfo{person}{Eric Budish}.}
  \bibinfo{year}{2010}\natexlab{}.
\newblock \showarticletitle{The combinatorial assignment problem: approximate
  competitive equilibrium from equal incomes}. In
  \bibinfo{booktitle}{\emph{{BQGT}}}. \bibinfo{publisher}{{ACM}},
  \bibinfo{pages}{74:1}.
\newblock


\bibitem[\protect\citeauthoryear{Friedman, Psomas, and Vardi}{Friedman
  et~al\mbox{.}}{2015}]%
        {DBLP:conf/sigecom/FriedmanPV15}
\bibfield{author}{\bibinfo{person}{Eric~J. Friedman},
  \bibinfo{person}{Christos{-}Alexandros Psomas}, {and} \bibinfo{person}{Shai
  Vardi}.} \bibinfo{year}{2015}\natexlab{}.
\newblock \showarticletitle{Dynamic Fair Division with Minimal Disruptions}. In
  \bibinfo{booktitle}{\emph{{EC}}}. \bibinfo{publisher}{{ACM}},
  \bibinfo{pages}{697--713}.
\newblock


\bibitem[\protect\citeauthoryear{Friedman, Psomas, and Vardi}{Friedman
  et~al\mbox{.}}{2017}]%
        {DBLP:conf/sigecom/FriedmanPV17}
\bibfield{author}{\bibinfo{person}{Eric~J. Friedman},
  \bibinfo{person}{Christos{-}Alexandros Psomas}, {and} \bibinfo{person}{Shai
  Vardi}.} \bibinfo{year}{2017}\natexlab{}.
\newblock \showarticletitle{Controlled Dynamic Fair Division}. In
  \bibinfo{booktitle}{\emph{{EC}}}. \bibinfo{publisher}{{ACM}},
  \bibinfo{pages}{461--478}.
\newblock


\bibitem[\protect\citeauthoryear{Garg and Murhekar}{Garg and Murhekar}{2022}]%
        {DBLP:journals/corr/abs-2204-14229}
\bibfield{author}{\bibinfo{person}{Jugal Garg} {and} \bibinfo{person}{Aniket
  Murhekar}.} \bibinfo{year}{2022}\natexlab{}.
\newblock \showarticletitle{Computing Pareto-Optimal and Almost Envy-Free
  Allocations of Indivisible Goods}.
\newblock \bibinfo{journal}{\emph{CoRR}}  \bibinfo{volume}{abs/2204.14229}
  (\bibinfo{year}{2022}).
\newblock


\bibitem[\protect\citeauthoryear{Goldman and Procaccia}{Goldman and
  Procaccia}{2014}]%
        {DBLP:journals/sigecom/GoldmanP14}
\bibfield{author}{\bibinfo{person}{Jonathan~R. Goldman} {and}
  \bibinfo{person}{Ariel~D. Procaccia}.} \bibinfo{year}{2014}\natexlab{}.
\newblock \showarticletitle{Spliddit: unleashing fair division algorithms}.
\newblock \bibinfo{journal}{\emph{SIGecom Exch.}} \bibinfo{volume}{13},
  \bibinfo{number}{2} (\bibinfo{year}{2014}), \bibinfo{pages}{41--46}.
\newblock


\bibitem[\protect\citeauthoryear{He, Procaccia, Psomas, and Zeng}{He
  et~al\mbox{.}}{2019}]%
        {DBLP:conf/ijcai/HePPZ19}
\bibfield{author}{\bibinfo{person}{Jiafan He}, \bibinfo{person}{Ariel~D.
  Procaccia}, \bibinfo{person}{Alexandros Psomas}, {and} \bibinfo{person}{David
  Zeng}.} \bibinfo{year}{2019}\natexlab{}.
\newblock \showarticletitle{Achieving a Fairer Future by Changing the Past}. In
  \bibinfo{booktitle}{\emph{{IJCAI}}}. \bibinfo{publisher}{ijcai.org},
  \bibinfo{pages}{343--349}.
\newblock


\bibitem[\protect\citeauthoryear{H{\"{o}}hne and van Stee}{H{\"{o}}hne and van
  Stee}{2021}]%
        {DBLP:journals/iandc/HohneS21}
\bibfield{author}{\bibinfo{person}{Felix H{\"{o}}hne} {and}
  \bibinfo{person}{Rob van Stee}.} \bibinfo{year}{2021}\natexlab{}.
\newblock \showarticletitle{Allocating contiguous blocks of indivisible chores
  fairly}.
\newblock \bibinfo{journal}{\emph{Inf. Comput.}}  \bibinfo{volume}{281}
  (\bibinfo{year}{2021}), \bibinfo{pages}{104739}.
\newblock


\bibitem[\protect\citeauthoryear{Igarashi}{Igarashi}{2022}]%
        {igarashi2022cut}
\bibfield{author}{\bibinfo{person}{Ayumi Igarashi}.}
  \bibinfo{year}{2022}\natexlab{}.
\newblock \bibinfo{title}{How to cut a discrete cake fairly}.
\newblock
\newblock
\showeprint[arxiv]{2209.01348}~[cs.GT]


\bibitem[\protect\citeauthoryear{Lipton, Markakis, Mossel, and Saberi}{Lipton
  et~al\mbox{.}}{2004}]%
        {DBLP:conf/sigecom/LiptonMMS04}
\bibfield{author}{\bibinfo{person}{Richard~J. Lipton},
  \bibinfo{person}{Evangelos Markakis}, \bibinfo{person}{Elchanan Mossel},
  {and} \bibinfo{person}{Amin Saberi}.} \bibinfo{year}{2004}\natexlab{}.
\newblock \showarticletitle{On approximately fair allocations of indivisible
  goods}. In \bibinfo{booktitle}{\emph{{EC}}}. \bibinfo{publisher}{{ACM}},
  \bibinfo{pages}{125--131}.
\newblock


\bibitem[\protect\citeauthoryear{Manurangsi and Suksompong}{Manurangsi and
  Suksompong}{2020}]%
        {DBLP:journals/siamdm/ManurangsiS20}
\bibfield{author}{\bibinfo{person}{Pasin Manurangsi} {and}
  \bibinfo{person}{Warut Suksompong}.} \bibinfo{year}{2020}\natexlab{}.
\newblock \showarticletitle{When Do Envy-Free Allocations Exist?}
\newblock \bibinfo{journal}{\emph{{SIAM} J. Discret. Math.}}
  \bibinfo{volume}{34}, \bibinfo{number}{3} (\bibinfo{year}{2020}),
  \bibinfo{pages}{1505--1521}.
\newblock


\bibitem[\protect\citeauthoryear{Manurangsi and Suksompong}{Manurangsi and
  Suksompong}{2021}]%
        {DBLP:journals/siamdm/ManurangsiS21}
\bibfield{author}{\bibinfo{person}{Pasin Manurangsi} {and}
  \bibinfo{person}{Warut Suksompong}.} \bibinfo{year}{2021}\natexlab{}.
\newblock \showarticletitle{Closing Gaps in Asymptotic Fair Division}.
\newblock \bibinfo{journal}{\emph{{SIAM} J. Discret. Math.}}
  \bibinfo{volume}{35}, \bibinfo{number}{2} (\bibinfo{year}{2021}),
  \bibinfo{pages}{668--706}.
\newblock


\bibitem[\protect\citeauthoryear{Misra, Sonar, Vaidyanathan, and Vaish}{Misra
  et~al\mbox{.}}{2021}]%
        {DBLP:journals/corr/abs-2101-09794}
\bibfield{author}{\bibinfo{person}{Neeldhara Misra}, \bibinfo{person}{Chinmay
  Sonar}, \bibinfo{person}{P.~R. Vaidyanathan}, {and} \bibinfo{person}{Rohit
  Vaish}.} \bibinfo{year}{2021}\natexlab{}.
\newblock \showarticletitle{Equitable Division of a Path}.
\newblock \bibinfo{journal}{\emph{CoRR}}  \bibinfo{volume}{abs/2101.09794}
  (\bibinfo{year}{2021}).
\newblock


\bibitem[\protect\citeauthoryear{Oh, Procaccia, and Suksompong}{Oh
  et~al\mbox{.}}{2021}]%
        {DBLP:journals/siamdm/OhPS21}
\bibfield{author}{\bibinfo{person}{Hoon Oh}, \bibinfo{person}{Ariel~D.
  Procaccia}, {and} \bibinfo{person}{Warut Suksompong}.}
  \bibinfo{year}{2021}\natexlab{}.
\newblock \showarticletitle{Fairly Allocating Many Goods with Few Queries}.
\newblock \bibinfo{journal}{\emph{{SIAM} J. Discret. Math.}}
  \bibinfo{volume}{35}, \bibinfo{number}{2} (\bibinfo{year}{2021}),
  \bibinfo{pages}{788--813}.
\newblock


\bibitem[\protect\citeauthoryear{Seidl}{Seidl}{2018}]%
        {DBLP:journals/jasss/Seidl18}
\bibfield{author}{\bibinfo{person}{Roman Seidl}.}
  \bibinfo{year}{2018}\natexlab{}.
\newblock \showarticletitle{Handbook of Computational Social Choice \emph{by
  Brandt Felix, Vincent Conitzer, Ulle Endriss, Jerome Lang, Ariel Procaccia}}.
\newblock \bibinfo{journal}{\emph{J. Artif. Soc. Soc. Simul.}}
  \bibinfo{volume}{21}, \bibinfo{number}{2} (\bibinfo{year}{2018}).
\newblock


\bibitem[\protect\citeauthoryear{Suksompong}{Suksompong}{2019}]%
        {DBLP:journals/dam/Suksompong19}
\bibfield{author}{\bibinfo{person}{Warut Suksompong}.}
  \bibinfo{year}{2019}\natexlab{}.
\newblock \showarticletitle{Fairly allocating contiguous blocks of indivisible
  items}.
\newblock \bibinfo{journal}{\emph{Discret. Appl. Math.}}  \bibinfo{volume}{260}
  (\bibinfo{year}{2019}), \bibinfo{pages}{227--236}.
\newblock


\bibitem[\protect\citeauthoryear{Zeng and Psomas}{Zeng and Psomas}{2020}]%
        {DBLP:conf/sigecom/ZengP20}
\bibfield{author}{\bibinfo{person}{David Zeng} {and}
  \bibinfo{person}{Alexandros Psomas}.} \bibinfo{year}{2020}\natexlab{}.
\newblock \showarticletitle{Fairness-Efficiency Tradeoffs in Dynamic Fair
  Division}. In \bibinfo{booktitle}{\emph{{EC}}}. \bibinfo{publisher}{{ACM}},
  \bibinfo{pages}{911--912}.
\newblock


\end{thebibliography}


\clearpage
\onecolumn
\appendix
\section{Two Agents and Mixed Manna}
\label{apx:two-agents-and-mixed-manna}

In this section, we allow valuations to be negative, which is often referred to as \textit{mixed manna}.
For two agents, we show that $O(T)$ adjustments are sufficient to achieve EF1 in Algorithm~\ref{envy-balancing-algorithm-for-mixed-manna}.
The algorithm is the extension of the Envy Balancing Algorithm for two agents with nonnegative valuations given by He et al. \cite{DBLP:conf/ijcai/HePPZ19}, which assumes that we know the information of all items in advance and requires no adjustments.

The idea of Algorithm~\ref{envy-balancing-algorithm-for-mixed-manna} is based on the envy-cycle elimination algorithm, which can be used to efficiently compute an EF1 allocation \cite{DBLP:conf/sigecom/LiptonMMS04}.
Recall that an allocation $A$ is \textit{envy-free (EF)} if $v_i(A_i) \geq v_i(A_j)$ for any $i, j \in N$.
In Algorithm~\ref{envy-balancing-algorithm-for-mixed-manna}, we maintain two disjoint allocations $G$ and $C$ satisfying that $G$ is EF and $C$ is EF1.
Since the combination of an EF allocation and an EF1 allocation is also EF1, the returned allocation, which is the combination of $G$ and $C$, is also EF1.
To this end, we only assign the new item to an agent in $C$ while maintaining the EF1 property of $C$ (Line~\ref{algoline:allocate1}--Line~\ref{algoline:allocate2}).
Whenever both agents envy each other in $C$, we swap their bundles in $C$ to make $C$ an EF allocation (Line~\ref{algoline:adjust1}--Line~\ref{algoline:adjust2}).
Once $C$ becomes EF, we merge $C$ into $G$ (Line~\ref{algoline:merge1}--Line~\ref{algoline:merge2}).

\begin{algorithm}[!htp]
        \caption{Envy Balancing Algorithm for Mixed Manna}
        \label{envy-balancing-algorithm-for-mixed-manna}
        \begin{algorithmic}[1]
            \Require {$v_1, v_2$}
            \State {$G \leftarrow (\emptyset, \emptyset), C \leftarrow (\emptyset, \emptyset)$}
            \For {$t = 1, \ldots, T$}
                \If {($a_1$ is unenvied in $C$ $\land$ $v_1(g_t) > 0$) $\lor$ ($v_1(g_t) > 0 \land v_2(g_t) \leq 0$) $\lor$ ($a_1$ does not envy $a_2$ in $C$ $\land$ $v_1(g_t) \leq 0$ $\land$ $v_2(g_t) \leq 0$)}\label{algoline:allocate1}
                    \State {$C \leftarrow (C_1 \cup \{g_t\}, C_2)$}
                    \Else \State {$C \leftarrow (C_1, C_2 \cup \{g_t\})$}\label{algoline:allocate2}
                \EndIf
                \If {$a_1$ and $a_2$ envy each other in $C$}\label{algoline:adjust1}
                    \State {$C \leftarrow (C_2, C_1)$}\label{algoline:adjust2}
                \EndIf
                \If {$C$ is envy-free}\label{algoline:merge1}
                    \State {$G \leftarrow (G_1 \cup C_1, G_2 \cup C_2)$}
                    \State {$C \leftarrow (\emptyset, \emptyset)$}\label{algoline:merge2}
                \EndIf
                \State {$A^t = (C_1 \cup G_1, C_2 \cup G_2)$}
            \EndFor
            \State {\Return $[A^1, A^2, \ldots, A^T]$}
        \end{algorithmic}
\end{algorithm}

\begin{theorem}
    Algorithm~\ref{envy-balancing-algorithm-for-mixed-manna} is an EF1 algorithm for two agents and mixed manna that requires $O(T)$ adjustments.
\end{theorem}

\begin{proof}
    We first prove that every returned allocation is EF1.
    Since the combination of an EF allocation and an EF1 allocation is EF1, it suffices to prove that $G$ is always EF and $C$ is always EF1.
    Notice that we merge $C$ into $G$ only if $C$ is EF, and the combination of two EF allocations is EF.
    Thus $G$ is always EF.
    
    Now we prove that $C$ is always EF1.
    We only need to show that, if $C$ is EF1 at the beginning of the current iteration, it is also EF1 after the iteration.
    If $C$ is EF1, at least one of the agents is unenvied in $C$ since otherwise the bundles in $C$ would have been swapped in the previous iteration, which makes $C$ an EF allocation.
    Thus there is also at least one agent who does not envy the other agent in $C$.
    After $g_t$ arrives, according to Line~\ref{algoline:allocate1}, if it is a good for an agent and a chore for the other agent, i.e., $v_1(g_t) > 0 \land v_2(g_t) \leq 0$ or $v_1(g_t) \leq 0 \land v_2(g_t) > 0$, it will be assigned to the agent who views it as a good, which ensures that both $v_1(C_1) - v_1(C_2)$ and $v_2(C_2) - v_2(C_1)$ don't decrease after the assignment of $g_t$.
    In this case, the EF1 property of $C$ is maintained.
    Next, if the new item is a good for both agents, i.e., $v_1(g_t) > 0 \land v_2(g_t) > 0$, the agent who it is assigned to is unenvied by the other agent in $C$.
    Thus by definition, $C$ is EF1 after the assignment of $g_t$.
    Finally, if the new item is a chore for both agents, i.e., $v_1(g_t) \leq 0 \land v_2(g_t) \leq 0$, the agent who it is assigned to does not envy the other agent in $C$.
    By definition, $C$ is EF1 after the assignment of $g_t$.
    
    It remains to prove the number of adjustments.
    Note that adjustments happen only when both agents envy each other in $C$, and, in this case, their bundles in $C$ are swapped, which contributes $|C_1 \cup C_2|$ adjustments.
    After the swap, $C$ becomes EF and is merged into $G$.
    Since no adjustments are made on $G$, each item is adjusted at most once.
    Therefore, the total number of adjustments is $O(T)$.
\end{proof}

\section{The Incentive Issue}
\label{apx:incentive-property-of-rr-alg}

It is natural to ask about the incentive issue.
That is, apart from the objective of fairness, we also hope that each agent is better off after the adjustments to make them happy with the new allocation.
It is not difficult to see that all the algorithms given in this paper in the noncontiguous setting satisfy the incentive property.
Furthermore, we prove in Lemma~\ref{lem:incentive-property-of-round-robin} that if we re-run the round-robin algorithm with a deterministic tie-breaking rule upon the arrival of a new item, the valuation of each agent will not decrease in the new allocation.
This immediately leads to the incentive property of the algorithms given by He et al. \cite{DBLP:conf/ijcai/HePPZ19} which are based on the round-robin algorithm.


Now we prove the incentive property of the round-robin algorithm.
Recall that in the round-robin algorithm, agents pick their favorite items alternately in each round.
Suppose that there are $n$ agents and $m$ items $g_1, \ldots, g_m$.
In the outcome of the round-robin algorithm, we use $C^j_i$ to denote the index of the item obtained by agent $i$ in round $j$.
In particular, $C^j_i = 0$ if agent $i$ does not obtain any item in round $j$, and let $v_i(0) = 0$ for all $i \in N$.
Note that there are $\lceil m / n \rceil$ rounds in total.
When a new item $g_{m + 1}$ arrives, we run the round-robin algorithm again to allocate all the $m + 1$ items.
Let $\widetilde{C}^j_i$ denote the new outcome analogously.

\begin{lemma}\label{lem:incentive-property-of-round-robin}
    For every $1 \leq j \leq \lceil (m + 1) / n \rceil$ and $i \in N$, we have $v_i(C^j_i) \leq v_i(\widetilde{C}^j_i)$.
\end{lemma}

\begin{proof}
    Without loss of generality, we assume that the agent with a smaller index always picks an item before the agent with a larger index in each round.
    Define
    \begin{align*}
        M_i^j = \{g_1, \ldots, g_m\} \setminus \{C_k^{\ell} \mid \text{either } \ell < j \text{ or } \ell = j, k < i\}
    \end{align*}
    as the set of remaining items before agent $i$ pick an item in round $j$, and
    \begin{align*}
        \widetilde{M}_i^j = \{g_1, \ldots, g_{m + 1}\} \setminus \{\widetilde{C}_k^{\ell} \mid \text{either } \ell < j \text{ or } \ell = j, k < i\}
    \end{align*}
    analogously.
    Since
    \begin{align*}
        C_i^j = \arg\max_{g \in M_i^j} v_i(g), \quad
        \widetilde{C}_i^j = \arg\max_{g \in \widetilde{M}_i^j} v_i(g),
    \end{align*}
    it suffices to show that $M_i^j \subseteq \widetilde{M}_i^j$, which trivially holds when $j = 0$.
    
    Assume that $i < n$ and the case of $i = n$ is treated similarly.
    Suppose for induction that $M_i^j \subseteq \widetilde{M}_i^j$ is satisfied.
    We show that $M_{i + 1}^j \subseteq \widetilde{M}_{i + 1}^j$.
    Since $|M_i^j| + 1 = |\widetilde{M}_i^j|$, let $g$ be the only item in $\widetilde{M}_i^j \setminus M_i^j$.
    On one hand, if $\widetilde{C}_i^j = g$, then $\widetilde{M}_{i + 1}^j = M_{i + 1}^j \cup \{C_i^j\}$.
    On the other hand, if $\widetilde{C}_i^j \neq g$, we have $\widetilde{C}_i^j = C_i^j$ and thus $\widetilde{M}_{i + 1}^j = M_{i + 1}^j \cup \{g\}$.
    Therefore, $M_{i + 1}^j \subseteq \widetilde{M}_{i + 1}^j$ holds.
\end{proof}

By contrast, with the contiguity constraint, suppose that the order of agents stays the same throughout the entire period, i.e., agent $i$ obtains the $i$-th block, which is the default assumption for identical valuations in this work.
In this case, the incentive property is not compatible with EF1, even for two agents with additive identical valuations.
For instance, there are three items with valuations of $1$, $3$ and $2$, respectively.
In round $2$, we have $A^2 = (\{g_1\}, \{g_2\})$.
However, in round $3$, the only contiguous allocation that satisfies EF1 is $A^3 = (\{g_1, g_2\}, \{g_3\})$, in which the valuation of agent $2$ is less than the previous round.

\section{Efficient Computation of Contiguous Leximinmin Allocations for Identical Valuations}\label{apx:dynamic-programming-to-compute-leximin2-allocations}

In this section, we discuss how to efficiently compute a contiguous leximinmin allocation.
Recall that the contiguous leximinmin allocation is the contiguous leximin allocation $A$ with lexicographically minimum $P(A)$ among all leximin allocations, where $P(A) = (\ell_0, \cdots, \ell_n)$ such that $\ell_0 = 0$ and $\ell_i$ is the index of the last item in $A_i$.
Given an instance with $n$ agents and $m$ items, we will apply Dynamic Programming to compute the contiguous leximinmin allocation in $O(n^2 m^2)$ time.

Suppose that there are $n$ agents with general identical valuations $v_1, \ldots, v_n$ and $m$ items $g_1, \ldots, g_m$ arranged on a line.
Recall that $M_{l, r} = \{g_{l + 1}, \ldots, g_r\}$ is the set of items whose indexes are in range $[l + 1, r]$.
In particular, let $M_{l, r} = \emptyset$ if $l \geq r$.
For an allocation $A = (A_1, \ldots, A_k)$, define $Q(A) = (v(A_{\sigma(1)}), v(A_{\sigma(2)}), \ldots, v(A_{\sigma(k)}))$ as the tuple obtained by sorting the valuations of all agents in nondecreasing order, where $\sigma$ is a permutation with length $k$ satisfying $v(A_{\sigma(i)}) \leq v(A_{\sigma(i + 1)})$ for any $i < k$.
Given allocations $A'$ and $A''$, we say that $A'$ is better than $A''$ if either $Q(A')$ is lexicographically larger than $Q(A'')$ or $Q(A') = Q(A'')$ and $P(A')$ is lexicographically smaller than $P(A'')$.
By definition, the contiguous leximinmin allocation is the best allocation among all contiguous allocations.

For any $i \in [n]$ and $j \in [m] \cup \{0\}$, let $A[i, j]$ be the leximinmin allocation for $M_j$ with $i$ agents.
Observe that if $A = (A_1, \ldots, A_i)$ is a leximinmin allocation, for any $k \in [i]$, $A' = (A_1, \ldots, A_k)$ is the leximinmin allocation for the first $|A_1| + \cdots + |A_k|$ items with $k$ agents.
Thus if the last block of $A[i, j]$ is $M_{k, j}$ for some $0 \leq k \leq j$, then $A[i, j] = A[i - 1, k] \circ M_{k, j}$, where we use $A \circ x$ to denote the tuple obtained by appending $x$ to the end of tuple $A$.
As a result, the recurrence relation of $A[i, j]$ is
\begin{align*}
    A[i, j] = 
    \begin{cases}
        (M_j), & i = 1,\\
        \text{the best allocation in } \{A[i - 1, k] \circ M_{k, j} \mid 0 \leq k \leq j\}, & \text{otherwise}.
    \end{cases}
\end{align*}

In the implementation, we only need to maintain $P(A[i, j])$ and $Q(A[i, j])$, which are the necessary information to compare two allocations, rather than recording the exact allocation.
We use $f[i, j]$ to denote $P(A[i, j])$ and $h[i, j]$ to denote $Q(A[i, j])$.
By the recurrence relation of $A[i, j]$, we have $f[i, j] = f[i - 1, k] \circ j$, where the last block of $A[i, j]$ is $M_{k, j}$.
Analogously, $h[i, j]$ can be derived by inserting $v(M_{k, j})$ to the appropriate position in $h[i - 1, k]$ to ensure the monotonicity of the elements in $h[i, j]$.

Now we analyze the time complexity of the Dynamic Programming approach that we discussed above.
For each pair $(i, j) \in [n] \times ([m] \cup \{0\})$, we need to enumerate all allocations in $\{ A[i - 1, k] \circ M_{k, j} \mid 0 \leq k \leq j\}$ and make comparisons among them to compute $A[i, j]$.
Thus the total number of comparisons to be made is $O(n m^2)$.
Moreover, each comparison can be finished in $O(n)$ time.
Therefore, the running time of the algorithm is $O(n^2 m^2)$.

\section{Missing Proofs of Section~\ref{sec:contiguous-setting-with-identical-valuations}}
\label{apx:missing-proofs-of-section-contiguous-setting-with-identical-valuations}

\thmPROPaalgorithmforidenticalvaluations*

\begin{proof}
    Let $B_t^{\max} = \max_{i \in [t]} B_i$.
    For a contiguous allocation $A = (A_1, \ldots, A_n)$, define $\ell(A_i)$ as the index of the last item in $A_i$.
    In particular, let $\ell(A_0) = 0$.
    We abuse the notation and let $\ell(A) = (\ell(A_1), \ldots, \ell(A_{n - 1}))$.
    For two contiguous allocations $A'$ and $A''$, we say that $\ell(A') \prec \ell(A'')$ if $\ell(A')$ is lexicographically smaller than $\ell(A'')$.
    We say that $\ell(A') \preceq \ell(A'')$ if $\ell(A') \prec \ell(A'')$ or $\ell(A') = \ell(A'')$.
    For any $t \in [T]$, define $X^t$ as the allocation with lexicographically minimum $\ell(A)$ among all contiguous allocations $A = (A_1, \ldots, A_n)$ for $M_t$ satisfying constraint factor $B_t^{\max}$, i.e., $\min_{i \in N} v(A_i) \geq B^{\max}_t$.
    We first show that $X^t$ is well-defined for all $t \in [T]$, and $\ell(X_i^t) \leq \ell(X_i^{t + 1})$ for all $t < T$ and $i < n$.
    Then we prove that $A^t = X^t$ for all $t \in [T]$.
    
    We first prove that $X^t$ is well-defined for all $t \in [T]$, which is equivalent to showing the existence of the contiguous allocation for $M_t$ that satisfies constraint factor $B_t^{\max}$.
    The existence of $X^t$ trivially holds for $t = 0$.
    We assume for induction that $X^t$ exists and show that $X^{t + 1}$ also exists.
    If $B_{t + 1}^{\max} = B_t^{\max}$, it is easy to show that by adding $g_{t + 1}$ to the last block of $X^t$, the corresponding allocation for $M_{t + 1}$ satisfies constraint factor $B_t^{\max} = B_{t + 1}^{\max}$ and we are done.
    If $B_{t + 1}^{\max} > B_t^{\max}$, since in this case we must have $B_{t + 1}^{\max} = B_{t + 1}$, by Lemma~\ref{proportional-lemma}, the allocation that satisfies constraint factor $B_{t + 1}$ exists.
    Therefore, $X^t$ is well-defined for all $t \in [T]$.

    Next, we show that for all $t < T$ and $i < n$, $\ell(X^t_i) \leq \ell(X^{t + 1}_i)$.
    Fix $t < T$.
    Suppose for contradiction that there exists $i < n$ such that $\ell(X^t_i) > \ell(X^{t + 1}_i)$ and $\ell(X^t_j) \leq \ell(X^{t + 1}_j)$ for all $j < i$.
    By the definition of $\ell(X^t_i)$, we have $X^{t + 1}_i \subsetneq X^t_i$.
    Since $B_1^{\max}, \ldots, B_T^{\max}$ are nondecreasing, we have $v(X_i^{t + 1}) \geq B^{t + 1} \geq B^t$ due to the definition of $X_i^{t + 1}$.
    In allocation $X^t$, by moving the items in $X^t_i \setminus X^{t + 1}_i$ from $X_i^t$ to $X_{i + 1}^t$, we obtain a new allocation $X$ for $M_t$ with $\ell(X) \prec \ell(X^t)$ that satisfies the constraint factor $B_t^{\max}$, which contradicts the lexicographic minimality of $X^t$.
    Thus $\ell(X_i^t) \leq \ell(X_i^{t + 1})$ for all $t < T$ and $i < n$.
    
    Now it suffices to prove that for all $t \in [T]$, $A^t = X^t$, which satisfies constraint factor $B_t^{\max} \geq B_t$ by definition.
    First, $A^t = X^t$ trivially holds for $t = 0$.
    Supposing for induction that $A^{t - 1} = X^{t - 1}$ where $t < T$, we will show that $A^t = X^t$.

    If $B_{t}^{\max} = B_{t - 1}^{\max}$, then $p_1, \ldots, p_{n-1}$ will not change in round $t$ since $A^{t - 1} = X^{t - 1}$ satisfies the constraint factor $B_{t - 1}^{\max}$, which indicates that $\ell(A^{t}) = \ell(A^{t - 1})$.
    By the lexicographic minimality of $X^{t}$, we have $\ell(X^{t}) \preceq \ell(A^{t})$.
    Thus,
    \begin{align*}
        \ell(X^{t - 1}) \preceq \ell(X^{t}) \preceq \ell(A^{t}) = \ell(A^{t - 1}) = \ell(X^{t - 1}),
    \end{align*}
    by which we conclude that $A^{t} = X^{t}$.
    
    Suppose $B_{t}^{\max} > B_{t - 1}^{\max}$, which implies that $B_{t} = B_{t}^{\max}$.
    At the beginning of round $t$, we have
    \begin{align*}
        p_i
        = \ell(A^{t - 1}_i)
        = \ell(X^{t - 1}_i)
        \leq \ell(X^{t}_i),
    \end{align*}
    for all $i < n$.
    It is trivial that $\ell(A_0^{t}) = \ell(X_0^{t})$.
    Assuming for induction that $\ell(A_{i - 1}^{t}) = \ell(X_{i - 1}^{t})$ where $i < n$, we will show that $\ell(A_i^{t}) = \ell(X_i^{t})$.
    On one hand, $\ell(A_i^{t}) \leq \ell(X_i^{t})$ must hold, since $v(X_i^{t}) \geq B_t$ and the condition of Line~\ref{algoline:move-pointer1} will be violated when $p_i = \ell(X_i^{t})$.
    On the other hand, $\ell(A_i^t) \geq \ell(X_i^t)$ must hold, since otherwise we have $A_i^t \subsetneq X_i^t$ and $v(A_i^t) \geq B_t$, and by moving the items in $X_i^t \setminus A_i^t$ from $X_i^t$ to $X_{i + 1}^t$, we obtain a new allocation $X$ for $M_t$ with $\ell(X) \prec \ell(X^t)$ that satisfies constraint factor $B_t = B_t^{\max}$, which contradicts the lexicographic minimality of $X^t$.
    Therefore, $A^t = X^t$ for all $t \in [T]$.
    
    Finally, we prove the number of adjustments required by Algorithm~\ref{online-proportional-algorithm-for-additive-identical-agents}.
    For an item $g_t$ that has arrived, the agent $i$ that $g_t$ belongs to satisfies $p_{i - 1} < t \leq p_i$.
    Since $p_i$ is nondecreasing for each $i \in N$, the number of adjustments made on any item $g_t$ is at most $n - 1$.
    As a result, the total number of adjustments is at most $(n-1)T = O(nT)$.
\end{proof}

\begin{figure}[t]
    \centering
    \subfigure{
		\begin{minipage}[t]{.5\textwidth}

\begin{tikzpicture}[thick,scale=0.6, every node/.style={transform shape}]
    \draw (0,0) rectangle (11,1);
    \filldraw[pattern color=red!70,pattern=north west lines] (0,0) rectangle (4,1);
    \draw (0,1.5) rectangle (10.5,2.5);
    \filldraw[pattern color=blue!70,pattern=north west lines] (0,1.5) rectangle (5,2.5);
    \draw (4, 0) -- (4, 1);
    \draw (8, 0) -- (8, 1);
    \draw (5, 1.5) -- (5, 2.5);
    \draw (7, 1.5) -- (7, 2.5);
    \node[font=\fontsize{20}{6}\selectfont] at (-1,2) {$X^t$};
    \node[font=\fontsize{20}{6}\selectfont] at (-1,0.5) {$X^{t + 1}$};
    
    \node[font=\fontsize{20}{6}\selectfont] at (7,3.1) {$\ell^1_k$};
    \node[font=\fontsize{20}{6}\selectfont] at (5,3.1) {$\ell^1_{k - 1}$};
    \node[font=\fontsize{20}{6}\selectfont] at (10.5,3.1) {$\ell^1_{n}$};
    \node[font=\fontsize{20}{6}\selectfont] at (0,3.1) {$\ell^1_{0}$};
    \node[font=\fontsize{30}{6}\selectfont] at (2.5,2) {$\ldots$};
    \node[font=\fontsize{30}{6}\selectfont] at (9,2) {$\ldots$};
    
    \node[font=\fontsize{20}{6}\selectfont] at (8,-0.6) {$\ell^2_k$};
    \node[font=\fontsize{20}{6}\selectfont] at (4,-0.6) {$\ell^2_{k - 1}$};
    \node[font=\fontsize{20}{6}\selectfont] at (11,-0.6) {$\ell^2_{n}$};
    \node[font=\fontsize{20}{6}\selectfont] at (0,-0.6) {$\ell^2_{0}$};
    \node[font=\fontsize{30}{6}\selectfont] at (2,0.5) {$\ldots$};
    \node[font=\fontsize{30}{6}\selectfont] at (9.5,0.5) {$\ldots$};
    
    \draw [decorate,decoration = {brace,amplitude=2mm},ultra thick] (0, 3.8) -- (7, 3.8);
    \node[font=\fontsize{20}{6}\selectfont] at (3.5,5) {$A^{11}$};
    
    \draw [decorate,decoration = {brace,mirror,amplitude=2mm},ultra thick] (0, -1.2) -- (8, -1.2);
    \node[font=\fontsize{20}{6}\selectfont] at (4,-2.2) {$A^{22}$};
\end{tikzpicture}

		\end{minipage}%
	}%
    \subfigure{
		\begin{minipage}[t]{.5\textwidth}

\begin{tikzpicture}[thick,scale=0.6, every node/.style={transform shape}]
    \draw (0,0) rectangle (11,1);
    \filldraw[pattern color=red!70,pattern=north west lines] (0,1.5) rectangle (4,2.5);
    \draw (0,1.5) rectangle (10.5,2.5);
    \filldraw[pattern color=blue!70,pattern=north west lines] (0,0) rectangle (5,1);
    \draw (4, 1.5) -- (4, 2.5);
    \draw (8, 0) -- (8, 1);
    \draw (5, 0) -- (5, 1);
    \draw (7, 1.5) -- (7, 2.5);
    
    \node[font=\fontsize{20}{6}\selectfont] at (7,3.1) {$\ell^1_k$};
    \node[font=\fontsize{20}{6}\selectfont] at (5,-0.6) {$\ell^1_{k - 1}$};
    \node[font=\fontsize{20}{6}\selectfont] at (10.5,3.1) {$\ell^1_{n}$};
    \node[font=\fontsize{20}{6}\selectfont] at (0,-0.6) {$\ell^1_{0}$};
    \node[font=\fontsize{30}{6}\selectfont] at (2,2) {$\ldots$};
    \node[font=\fontsize{30}{6}\selectfont] at (9,2) {$\ldots$};
    
    \node[font=\fontsize{20}{6}\selectfont] at (8,-0.6) {$\ell^2_k$};
    \node[font=\fontsize{20}{6}\selectfont] at (4,3.1) {$\ell^2_{k - 1}$};
    \node[font=\fontsize{20}{6}\selectfont] at (11,-0.6) {$\ell^2_{n}$};
    \node[font=\fontsize{20}{6}\selectfont] at (0,3.1) {$\ell^2_{0}$};
    \node[font=\fontsize{30}{6}\selectfont] at (2.5,0.5) {$\ldots$};
    \node[font=\fontsize{30}{6}\selectfont] at (9.5,0.5) {$\ldots$};
    
    \draw [decorate,decoration = {brace,amplitude=2mm},ultra thick] (0, 3.8) -- (7, 3.8);
    \node[font=\fontsize{20}{6}\selectfont] at (3.5,5) {$A^{21}$};
    
    \draw [decorate,decoration = {brace,mirror,amplitude=2mm},ultra thick] (0, -1.2) -- (8, -1.2);
    \node[font=\fontsize{20}{6}\selectfont] at (4,-2.2) {$A^{12}$};
\end{tikzpicture}
		\end{minipage}%
	}%
    \caption{Figure illustrating the proof of Lemma~\ref{lem:nondecreasing-leximin-allocations}}
    \label{fig:demon-for-leximin}
\end{figure}

\lemleximinallocationswithboundedvaluations*

\begin{proof}
    Let's first specify some conventions.
    Since for any feasible tuple $O = (\ell_0, \ldots, \ell_n)$, there exists the only allocation $A$ satisfying $P(A) = O$, from now on we will use a tuple to refer to the corresponding allocation.
    For an allocation $A = (A_1, \ldots, A_k)$ and $O = P(A)$, define $Q(O) = (v(A_{\sigma(1)}), v(A_{\sigma(2)}), \ldots, v(A_{\sigma(k)}))$ as the tuple obtained by sorting the valuations of all agents in nondecreasing order, where $\sigma$ is a permutation with length $k$ satisfying $v(A_{\sigma(i)}) \leq v(A_{\sigma(i + 1)})$ for any $i < k$.
    For two tuples $O_1$ and $O_2$, we say $O_1 \preceq O_2$ if $O_1$ is lexicographically not larger than $O_2$.
    Without loss of generality, we assume that $v(g) > 0$ for all $g \in M$.
        
    Fix round $t$.
    We will prove that $P_j(X^t) \leq P_j(X^{t + 1})$ for every $j \in N$.
    Let $P_j(X^t) = (\ell_0^1, \ell_1^1, \ldots, \ell_n^1)$ and $P_j(X^{t + 1}) = (\ell_0^2, \ell_1^2, \ldots, \ell_n^2)$, as illustrated in Figure~\ref{fig:demon-for-leximin}.
    Since both $X^t$ and $X^{t + 1}$ are leximinmin, for any $j \in [n]$ and $k \in \{1, 2\}$, $(\ell_0^k, \ldots, \ell_j^k)$ must be a leximinmin allocation for $M_{\ell^k_j}$ with $j$ agents.
    Thus if $\ell_j^1 = \ell_j^2$ for some $j \in N$, then $\ell_k^1 = \ell_k^2$ for all $k \leq j$.
    
    Suppose for contradiction that there exists $j \in N$ such that $\ell_j^1 > \ell_j^2$.
    Let $k$ be the first index larger than $j$ such that $\ell^1_k < \ell^2_k$.
    Since $\ell^1_n = t < t + 1 = \ell^2_n$, such $k$ must exist and, in the same time, satisfies that $\ell_{k-1}^1 > \ell_{k-1}^2$.
    For $p, q \in \{1, 2\}$, define $A^{pq}$ as the allocation for $M_{\ell_k^q}$ such that $P(A^{pq}) = (\ell_0^p, \ell_1^p, \ldots, \ell_{k - 1}^p, \ell_k^q)$.
    Since both $A^{11}$ and $A^{22}$ are leximinmin, we have
    \begin{align*}
        Q(A^{21}) \preceq Q(A^{11}),
    \end{align*}
    and
    \begin{align}
        Q(A^{12}) \preceq Q(A^{22}).\label{eqn:A12-leq-A22}
    \end{align}
    It suffices to show that $Q(A^{11}) \preceq Q(A^{21})$ and $Q(A^{22}) \preceq Q(A^{12})$.
    Since in this case, if 
    \begin{align*}
        (\ell_0^2, \ell_1^2, \ldots, \ell_{k - 1}^2) \preceq (\ell_0^1, \ell_1^1, \ldots, \ell_{k - 1}^1),
    \end{align*}
    we can replace the first $k$ elements of $P(X^t) = (\ell_0^1, \ldots, \ell_n^1)$ with $\ell_0^2, \ldots, \ell_{k - 1}^2$ to obtain another leximin allocation $A$ with $P(A) \preceq P(X^t)$, which contradicts the fact that $X^t$ is leximinmin.
    By contrast, if 
    \begin{align*}
        (\ell_0^1, \ell_1^1, \ldots, \ell_{k - 1}^1) \preceq (\ell_0^2, \ell_1^2, \ldots, \ell_{k - 1}^2),
    \end{align*}
    we can replace the first $k$ elements of $P(X^{t + 1}) = (\ell_0^2, \ldots, \ell_n^2)$ with $\ell_0^1, \ldots, \ell_{k - 1}^1$ to obtain another leximin allocation $A$ with $P(A) \preceq P(X^{t + 1})$, which contradicts the fact that $X^{t + 1}$ is leximinmin.
    
    Due to the symmetry, we only prove $Q(A^{11}) \preceq Q(A^{21})$.
    Since $\ell_k^1 < \ell_k^2$ and $\ell_{k - 1}^1 > \ell_{k - 1}^2$, we have
    \begin{align}
        v(A_k^{11}) < v(A_k^{12}) < v(A_k^{22}),\label{eqn:Ak11<Ak12<Ak22}
    \end{align}
    and thus
    \begin{align}
        v(A_k^{11}) = v(A_k^{12}) - v(M_{\ell_k^2} \setminus M_{\ell_k^1}) < v(A_k^{22}) - v(M_{\ell_k^2} \setminus M_{\ell_k^1}) = v(A_k^{21}).\label{eqn:Ak11<Ak21}
    \end{align}
    Note that all but the last blocks in $A^{12}$ are the same as those in $A^{11}$, and all but the last blocks in $A^{22}$ are the same as those in $A^{21}$.
    Recall that when deciding the lexicographic order of two tuples, we first compare the first elements of two tuples; if they are equal, then compare the second ones, and so on.
    If \eqref{eqn:A12-leq-A22} is decided only by the blocks in $A^{12}$ and $A^{22}$ with valuations smaller than $v(A_k^{11})$, then $Q(A^{11}) \preceq Q(A^{21})$ holds because the blocks with valuations smaller than $v(A_k^{11})$ in $A^{11}$ and $A^{21}$ are the same with those in $A^{12}$ and $A^{22}$.
    On the other hand, if \eqref{eqn:A12-leq-A22} is decided by the blocks with valuations at least $v(A_k^{11})$, \eqref{eqn:A12-leq-A22} implies that the number of blocks with valuations $v(A_k^{11})$ in $A^{12}$ is not smaller than the number of blocks with valuations $v(A_k^{11})$ in $A^{22}$.
    By \eqref{eqn:Ak11<Ak12<Ak22} and \eqref{eqn:Ak11<Ak21}, the number of blocks with valuations $v(A_k^{11})$ in $A^{11}$ is strictly larger than the number of blocks with valuations $v(A_k^{11})$ in $A^{21}$.
    Therefore, $Q(A^{11}) \preceq Q(A^{21})$ is established.
\end{proof}

\section{Missing Proofs of Section~\ref{sec:contiguous-setting-with-nonidentical-valuations}}
\label{apx:missing-proofs-of-contiguous-nonidentical-setting}

\thmnonidenticallowerboundforEF*

\begin{proof}
    Suppose that there are $T$ items and $n$ agents with valuation $v_1, \ldots, v_n$.
    For any $t \in [T]$, let
    \begin{align*}
        v_i(g_t) = 
        \begin{cases}
            n^{2c}, & 3cn + 3(i - 1) < t \leq 3cn + 3i \text{ for some } c \geq 0,\\
            0, & \text{otherwise},
        \end{cases}
    \end{align*}
    for all $i \in N$.
    That is, each period of length $3 n$ is divided into $n$ blocks of length $3$, and agent $i$ is only interested in the items in the $i$-th block in each period $c$ with $n^{2c}$ valuation for each of these items.
    It suffices to show that, for every $k \geq 2n$, the first $3k - 3n$ items belong to one agent in round $3k$ and another agent in round $3k + 3$.
    Since it follows that the number of adjustments required is at least $\sum_{k=2n}^{T/3 - 1} 3(k - n) = \Theta(T^2)$, where we assume $T \gg n$.
     
    We only give the proof for $k$ such that $k$ is a multiple of $n$, and the proof can be easily generalized to any $k \geq 2n$.
    Now we prove that, for any $c \geq 2$ and $k = cn$, the first $3k - 3n = (c - 1) \cdot 3n$ items belong to agent $1$ in round $3k$ and belong to agent $2$ in round $3k + 3$.
    Define $G_i = \{g_t \mid (c - 1) \cdot 3n + 3(i - 1) < t \leq (c - 1) \cdot 3n + 3i\}$ as the set of items that arrive during period $c$ and agent $i$ is interested in.
    In round $3k$, if agent $i$ gets none of the items in $G_i$, there must exist another agent $j$ that obtains at least two of them due to the contiguity requirement.
    Furthermore, the total valuation of the first $3k - 3n = (c - 1) \cdot 3n$ items for agent $i$ is
    \begin{align*}
        3 \sum_{j=0}^{c - 2} n^{2j} = 3 \cdot \frac{n^{2(c - 1)} - 1}{n^2 - 1} < n^{2(c - 1)} = v_i(g), \quad \forall g \in G_i.
    \end{align*}
    Thus if agent $i$ gets none of the items in $G_i$, he must envy agent $j$ up to one item.
    As a result, to satisfy EF1, each agent $i$ should get at least one of the items in $G_i$.
    Due to the contiguity requirement, for every $i \in N$, agent $i$ must get the $i$-th block which contains at least one item in $G_i$.
    In this case, the first $3k - 3n = (c - 1) \cdot 3n$ items belong to agent $1$.
    Similarly, we can show that in round $3k + 3$, the first $3k - 3n + 3$ items belong to agent $2$ and we are done.
\end{proof}

\thmlowerboundforbinaryagents*

\begin{proof}
    \begin{figure}[t]
        \centering

\begin{tikzpicture}
    \node[font=\fontsize{15}{6}\selectfont] at (-1.5,0) {$v_2(g_t)$};
    \node[font=\fontsize{15}{6}\selectfont] at (-1.5,1) {$v_1(g_t)$};
    \node[font=\fontsize{15}{6}\selectfont] at (0,0) {$1$};
    \node[font=\fontsize{15}{6}\selectfont] at (0,1) {$1$};
    \node[font=\fontsize{20}{6}\selectfont] at (1,0) {$\ldots$};
    \node[font=\fontsize{20}{6}\selectfont] at (1,1) {$\ldots$};
    \node[font=\fontsize{15}{6}\selectfont] at (2,0) {$1$};
    \node[font=\fontsize{15}{6}\selectfont] at (2,1) {$1$};
    
    \draw [decorate,decoration = {brace,amplitude=2mm},ultra thick] (-0.1, 1.5) -- (2.1, 1.5);
    \node[font=\fontsize{15}{6}\selectfont] at (1,2.3) {$2T + 2$};
    
    \node[font=\fontsize{15}{6}\selectfont] at (3,0) {$0$};
    \node[font=\fontsize{15}{6}\selectfont] at (3,1) {$1$};
    \node[font=\fontsize{15}{6}\selectfont] at (3.5,0) {$0$};
    \node[font=\fontsize{15}{6}\selectfont] at (3.5,1) {$1$};
    
    \node[font=\fontsize{15}{6}\selectfont] at (4.5,0) {$1$};
    \node[font=\fontsize{15}{6}\selectfont] at (4.5,1) {$0$};
    \node[font=\fontsize{15}{6}\selectfont] at (5,0) {$1$};
    \node[font=\fontsize{15}{6}\selectfont] at (5,1) {$0$};
    \node[font=\fontsize{15}{6}\selectfont] at (5.5,0) {$1$};
    \node[font=\fontsize{15}{6}\selectfont] at (5.5,1) {$0$};
    \node[font=\fontsize{15}{6}\selectfont] at (6,0) {$1$};
    \node[font=\fontsize{15}{6}\selectfont] at (6,1) {$0$};
    \node[font=\fontsize{15}{6}\selectfont] at (6.5,0) {$0$};
    \node[font=\fontsize{15}{6}\selectfont] at (6.5,1) {$1$};
    \node[font=\fontsize{15}{6}\selectfont] at (7,0) {$0$};
    \node[font=\fontsize{15}{6}\selectfont] at (7,1) {$1$};
    \node[font=\fontsize{15}{6}\selectfont] at (7.5,0) {$0$};
    \node[font=\fontsize{15}{6}\selectfont] at (7.5,1) {$1$};
    \node[font=\fontsize{15}{6}\selectfont] at (8,0) {$0$};
    \node[font=\fontsize{15}{6}\selectfont] at (8,1) {$1$};
    
    \node[font=\fontsize{15}{6}\selectfont] at (9,0) {$1$};
    \node[font=\fontsize{15}{6}\selectfont] at (9,1) {$0$};
    \node[font=\fontsize{15}{6}\selectfont] at (9.5,0) {$1$};
    \node[font=\fontsize{15}{6}\selectfont] at (9.5,1) {$0$};
    \node[font=\fontsize{15}{6}\selectfont] at (10,0) {$1$};
    \node[font=\fontsize{15}{6}\selectfont] at (10,1) {$0$};
    \node[font=\fontsize{15}{6}\selectfont] at (10.5,0) {$1$};
    \node[font=\fontsize{15}{6}\selectfont] at (10.5,1) {$0$};
    \node[font=\fontsize{15}{6}\selectfont] at (11,0) {$0$};
    \node[font=\fontsize{15}{6}\selectfont] at (11,1) {$1$};
    \node[font=\fontsize{15}{6}\selectfont] at (11.5,0) {$0$};
    \node[font=\fontsize{15}{6}\selectfont] at (11.5,1) {$1$};
    \node[font=\fontsize{15}{6}\selectfont] at (12,0) {$0$};
    \node[font=\fontsize{15}{6}\selectfont] at (12,1) {$1$};
    \node[font=\fontsize{15}{6}\selectfont] at (12.5,0) {$0$};
    \node[font=\fontsize{15}{6}\selectfont] at (12.5,1) {$1$};
    
    \node[font=\fontsize{20}{6}\selectfont] at (13.5,0) {$\ldots$};
    \node[font=\fontsize{20}{6}\selectfont] at (13.5,1) {$\ldots$};
    
    \draw [decorate,decoration = {brace,amplitude=2mm},ultra thick] (4.4, 1.5) -- (8.1, 1.5);
    \node[font=\fontsize{15}{6}\selectfont] at (6.5,2.3) {Period $1$};
    
    \draw [decorate,decoration = {brace,amplitude=2mm},ultra thick] (8.9, 1.5) -- (12.6, 1.5);
    \node[font=\fontsize{15}{6}\selectfont] at (11,2.3) {Period $2$};
\end{tikzpicture}
        \caption{Figure illustrating the instance in the proof of Theorem~\ref{thm:lower-bound-for-2-binary-agents}}
        \label{fig:lower-bound-of-binary-valuations}
    \end{figure}

    We describe an instance with $4T+4$ items where $T > 0$ is a multiple of $4$.
    The instance is illustrated in Figure~\ref{fig:lower-bound-of-binary-valuations}.
    We say that an item $g$ is type $0$ if $v_1(g) = v_2(g) = 1$, type $1$ if $v_1(g) = 1, v_2(g) = 0$, and type $2$ if $v_1(g) = 0, v_2(g) = 1$.
    The first $2T + 2$ items are type $0$.
    The following $2$ items are type $1$.
    The remaining $2T$ items are divided into periods with a length of $8$.
    In each period, the first $4$ items are type $2$ and the last $4$ items are type $1$.
    For any $0 \leq k < T / 4$, we show that the first block must belong to agent $2$ in round $2T + 8k + 4$ and must belong to agent $1$ in round $2T + 8k + 8$.
    
    Given an allocation, we say that an item is non-wasteful if it has a valuation of $1$ for the agent that obtains it.
    In round $2T + 8k + 4$, the numbers of non-wasteful items required by agent $1$ and agent $2$ are $T + 2k + 2$ and $T + 2k + 1$, respectively.
    Thus an EF1 allocation should contain at least $2T + 4k + 3$ non-wasteful items.
    Suppose for contradiction that the second block belongs to agent $2$.
    Denote $G$ as the set of the last $8k + 2$ items.
    Since $v_2(G) = 4k < T + 2k + 1$, all items in $G$ must belong to agent $2$, and thus there are $|G| - v_2(G) = 4k + 2$ wasteful items in $G$.
    As a result, the number of non-wasteful items is at most
    \begin{align*}
        (2T + 8k + 4) - (4k + 2) = 4T + 4k + 2 < 2T + 4k + 3,
    \end{align*}
    which leads to a contradiction.
    Therefore, the first block must belong to agent $2$ in round $2T + 8k + 4$.
    
    The proof for round $2T + 8k + 8$ is analogous.
    In round $2T + 8k + 8$, the numbers of non-wasteful items required by agent $1$ and agent $2$ are $T + 2k + 2$ and $T + 2k + 3$, respectively.
    Thus an EF1 allocation should contain at least $2T + 4k + 5$ non-wasteful items.
    Suppose for contradiction that the second block belongs to agent $1$.
    Denote $G$ as the set of the last $8k + 6$ items.
    Since $v_1(G) = 4k + 2 < T + 2k + 3$, all items in $G$ must belong to agent $1$, and thus there are $|G| - v_1(G) = 4k + 4$ wasteful items in $G$.
    As a result, the number of non-wasteful items is at most
    \begin{align*}
        (2T + 8k + 8) - (4k + 4) = 2T + 4k + 4 < 2T + 4k + 5,
    \end{align*}
    which leads to a contradiction.
    Therefore, the first block must belong to agent $1$ in round $2T + 8k + 8$.
    
    Finally, we prove the lower bound of the number of adjustments.
    Start from round $2T + 4$ and the agent who obtains the first block alternates after every $4$ rounds.
    Since the first block must contain the first $T$ items, the number of adjustments is at least $T^2 / 2 = \Theta(T^2)$.
\end{proof}

Unfortunately, the instance given above only provides a $\Omega(nT)$ lower bound after being generalized to any number of agents, which is directly implied by Theorem~\ref{thm:lower-bound-for-identical-valuations-and-proportionality}.

\end{document}